\documentclass[12pt]{article}
\usepackage[utf8]{inputenc}
\usepackage{authblk}
\usepackage{setspace}
\usepackage[margin=1.0in]{geometry}
\usepackage{graphicx}
\usepackage{subcaption}
\usepackage{lineno}
\usepackage{url}            
\usepackage{lastpage}
\usepackage{amsmath, amssymb, amsthm, amsfonts}
\usepackage{color}
\usepackage{caption}
\usepackage{mathtools}
\usepackage{mathrsfs} 
\usepackage[algoruled]{algorithm2e}
\usepackage{natbib}
\usepackage{enumitem}
\usepackage[colorlinks=true, urlcolor=blue, linkcolor=blue, citecolor=blue]{hyperref}
\usepackage{stmaryrd}
\usepackage{graphbox}   

\DeclareMathOperator{\rank}{rank}

\DeclareMathOperator{\krank}{kr}
\DeclareMathOperator{\tr}{tr}

\DeclareMathOperator*{\vectorize}{vec}
\DeclareMathOperator*{\diag}{diag}

\DeclareMathOperator*{\argmin}{argmin}

\newtheorem{proposition}{Proposition}
\newtheorem{corollary}{Corollary}

\newtheorem{lemma}{Lemma}

\providecommand{\keywords}[1]{\textbf{\textit{Keywords:}} #1}

\title{\textbf{Latent Functional PARAFAC for modeling multidimensional longitudinal data}}

\author[1]{Lucas Sort}
\author[1]{Laurent Le Brusquet}
\author[1]{Arthur Tenenhaus}

\affil[1]{\textit{Université Paris-Saclay, CNRS, CentraleSupélec, Laboratoire des Signaux et Systèmes, Gif-sur-Yvette, 91190 France}}
\affil[ ]{\textit{\texttt{lucas.sort@centralesupelec.fr}}}

\date{}

\begin{document}

\maketitle

\begin{abstract}
In numerous settings, it is increasingly common to deal with longitudinal data organized as high-dimensional multi-dimensional arrays, also known as tensors. Within this framework, the time-continuous property of longitudinal data often implies a smooth functional structure on one of the tensor modes. To help researchers investigate such data, we introduce a new tensor decomposition approach based on the CANDECOMP/PARAFAC decomposition. Our approach allows for representing a high-dimensional functional tensor as a low-dimensional set of functions and feature matrices. Furthermore, to capture the underlying randomness of the statistical setting more efficiently, we introduce a probabilistic latent model in the decomposition. A covariance-based block-relaxation algorithm is derived to obtain estimates of model parameters. Thanks to the covariance formulation of the solving procedure and thanks to the probabilistic modeling, the method can be used in sparse and irregular sampling schemes, making it applicable in numerous settings. We apply our approach to help characterize multiple neurocognitive scores observed over time in the Alzheimer’s Disease Neuroimaging Initiative (ADNI) study. Finally, intensive simulations show a notable advantage of our method in reconstructing tensors.  
\end{abstract}


\newpage

\section{Introduction}
\label{sec:introduction}

When handling high-dimensional tensor data, it is common to seek a low-dimensional representation of the data to capture the latent information more efficiently and enhance characterization. Indeed, This low-dimensional representation often helps reduce the data's high-dimensionality-induced complexity. In this setting, it is increasingly common to manipulate tensors with some underlying smooth structure in one of the modes. Additionally, since data is often collected through a sampling process, tensors commonly feature a mode (the "sample mode") on which the induced sub-tensors can be seen as samples drawn from a random tensor distribution. In longitudinal studies involving tensor-valued data, these two properties are often encountered: the time-continuous property naturally implies that one of the modes has a smooth structure, while the sampling setting introduces randomness along one mode of the tensor, as described previously. Since integrating these smooth and probabilistic characteristics would lead to more interpretable results, extending the functional and multiway data analysis procedures appears essential to leverage the increasing complexity of their datasets.

In the tensor literature, the CANDECOMP/PARAFAC (CP) decomposition (\cite{Harshman1970}, \cite{Carroll1970}), also known as the Canonical Polyadic Decomposition (CPD), is a well-known dimension reduction approach used to represent a tensor as a sum of rank-1 tensors. The method effectively reduces the high-dimensionality of a tensor to a collection of matrices with a matching number of columns $R$, also called the rank. The Tucker decomposition (\cite{Tucker1966}) is a closely related technique reducing the dimension of a tensor to a collection of matrices, with a varying number of columns $r_1, \dots, r_D$, and a (smaller) tensor, of dimensions $r_1 \times \dots \times r_D$, called the tensor core. Both approaches have been studied intensively and applied to numerous statistical problems (\cite{Sidiropoulos2017}), including regression (\cite{Zhou2013}, \cite{Lock2018b}). Finally, to account for potential randomness in a tensor, \cite{Lock2018a} introduced a probabilistic model in CPD on one of the modes. Their approach goes further by allowing the integration of auxiliary covariates in the probabilistic distribution, paving the way to supervised tensor decomposition.

Over the past decades, numerous methods have been introduced to investigate data with some underlying smooth structure. In the literature, such data is often referred to as functional data (\cite{Ramsay2005},  \cite{Wang2016}). In this context, the Functional Principal Component Analysis (FPCA) is a well-known method used to reduce the underlying infinite-dimensional space of functional data to the low-dimensional space of the most informative functional modes of variation, given by the eigenfunctions of the covariance operator. (\cite{Rice1991}) introduced a smoothing penalty in the mean and covariance estimations, allowing the retrieval of smooth eigenfunctions. In their seminal paper, \cite{Yao2005} proposed using a weighted linear least square procedure to estimate the covariance surface and obtain smooth eigenfunctions. Various extensions of the FPCA were proposed, notably for multilevel data (\cite{Di2009}) and multivariate data (\cite{Happ2018}). Similarly to the tensor setting, these developments were applied to other statistical problems, like regression (\cite{Morris2015}).

Conversely, little has been proposed to investigate tensors with a smooth structure. In this context, most methods rely on smoothness penalties, notably for tensor decomposition (\cite{Yokota2016}, \cite{imaizumi17a}). Extending the supervised CPD setting of \cite{Lock2018a}, \cite{Guan2023} proposed an interesting alternative approach to decompose an order-3 smooth tensor using a smoothness penalty. A significant advantage of this approach is that it is robust to sparse sampling schemes, as commonly encountered in longitudinal settings. Using a different approach, \cite{Han2023} extended the functional singular value decomposition to tensor-valued data with a Reproducing-Kernel Hilbert Spaces (RKHS) methodology. However, the existing literature lacks a general framework extending CPD to functional tensors of any order, capable of handling sparse and irregular sampling while accommodating the inherent randomness in sampling. Indeed, although we think it could be adapted, the approach proposed by \cite{Guan2023} does not support irregular sampling and is limited to order-3 tensors. Additionally, the method introduced in \cite{Han2023} lacks probabilistic modeling and also cannot handle irregular schemes, a substantial limitation in numerous longitudinal settings. 

In this context, we introduce a new probabilistic and functional PARAFAC model for functional tensors of any order with an underlying probabilistic structure. Our method uses a covariance-based algorithm that requires the estimation of (cross-)covariance surfaces of the functional entries of the tensor, allowing the decomposition of highly sparse and irregular data. Furthermore, we introduce a Bayesian approach, inspired by \cite{Yao2005}, to predict individual-mode factors and fit the data more accurately under mild assumptions.

The paper is organized as follows. In Section \ref{section:model}, we introduce notations and existing methods before presenting the proposed method. In Section \ref{section:estimation} we provide the estimation setting and introduce the solving procedure for the method. An estimation method for sample-mode vectors is shown afterward. Then, we propose to test our method in Section \ref{section:application} to help characterize multiple neurocognitive scores observed over time in the Alzheimer’s Disease Neuroimaging Initiative study (ADNI) study. In Section \ref{section:simulations}, we present results obtained on simulated data. Finally, we discuss our method's limitations and possible extensions in Section \ref{section:discussion}. All proofs are given in Appendix \ref{sec:annex:proofs}. 

\section{Model}
\label{section:model}

\subsection{Notations}

\subsubsection{Tensor notations}

In the following, we denote a vector by a \textit{small} bold letter $\bold{a}$, a matrix by a \textit{capital} bold letter $\bold{A}$, and a tensor by a \textit{calligraphic} letter $\mathcal{A}$. Scalars are denoted so they do not fall in any those categories. Additionally, we denote $[D] = \{1, \dots, D\}$. Elements of a vector $\bold{a}$ are denoted $a_j$, of a matrix $\bold{A}$, $a_{jj'}$, and of a tensor $\mathcal{A}$ of order $D$, $a_{j_1 \dots j_D}$, often abbreviated using the vector of indices $\bold{j} = (j_1, \dots, j_D)$ as $a_{\bold{j}}$. The $j$th column of $\bold{A}$ is denoted $\bold{a}_j$. The Frobenius norm of a matrix is defined as $\|\bold{A}\|_F = (\tr(\bold{A}^\top \bold{A}))^{1/2} = (\sum_{j, j'} (a_{jj'})^2)^{1/2}$. We extend this definition to tensors as $\|\mathcal{A}\|_F = (\sum_{\bold{j}} (a_{\bold{j}})^2)^{1/2}$. We denote the mode-$d$ matricization of a tensor $\mathcal{X}$ as the matrix $\bold{X}_{(d)} \in \mathbb{R}^{p_d \times p_{(-d)}}$, where $p_{(-d)} = \prod_{d'\neq d} p_{d'}$, which columns correspond to mode-$d$ fibers of the tensor $\mathcal{X}$. We denote the vectorization operator, stacking the columns of a matrix $\bold{A}$ into a column vector, as $\vectorize(\bold{A})$. In this context, we define the mode-$d$ vectorization as $\bold{x}_{(d)} = \vectorize(\bold{X}_{(d)}) \in \mathbb{R}^{P\times 1}$, where $P = \prod_{d'} p_{d'}$. The mode-$1$ vectorization of $\mathcal{X}sor$ is sometimes simply denoted $\bold{x} = \bold{x}_{(1)}$. For two matrices $\bold{A} \in \mathbb{R}^{p\times d}$ and $\bold{B} \in \mathbb{R}^{p' \times d'}$, the Kronecker product between $\bold{A}$ and $\bold{B}$ is defined as
\begin{equation*}
    \bold{A} \otimes \bold{B} = \begin{bmatrix}
        a_{11} \bold{B} &  \dots & a_{1d} \bold{B} \\
        \vdots & \ddots & \vdots \\
        a_{p1} \bold{B} &  \dots & a_{pd} \bold{B}
    \end{bmatrix} \in \mathbb{R}^{pp'\times dd'}
\end{equation*}
The Khatri-Rao product between two matrices with fixed number of columns $\bold{A} = [\bold{a}_1 \dots \bold{a}_d] \in \mathbb{R}^{p\times d}$ and $\bold{B} = [\bold{b}_1 \dots \bold{b}_d]\in \mathbb{R}^{p' \times d}$, is defined as a column-wise Kronecker product:
\begin{equation*}
    \bold{A} \odot \bold{B} = \begin{bmatrix}
        \bold{a}_{1} \otimes \bold{b}_{1} &  \dots & \bold{a}_{d} \otimes \bold{b}_{d}
    \end{bmatrix} \in \mathbb{R}^{pp'\times d}
\end{equation*}
Finally, the Hadamard product between two matrices with similar dimensions $\bold{A} \in \mathbb{R}^{p\times d}$ and $\bold{B} \in \mathbb{R}^{p\times d}$ is defined as the element-wise product of the two matrices:
\begin{equation*}
    \bold{A} * \bold{B} = \begin{bmatrix}
    a_{11}b_{11} &\dots &a_{1d}b_{1d}\\
    \vdots &\ddots &\vdots\\
    a_{p1}b_{p1} &\dots &a_{pd}b_{pd}
    \end{bmatrix} \in \mathbb{R}^{p\times d}
\end{equation*}

\subsubsection{Functional notations}

Considering an interval $\mathcal{I} \subset \mathbb{R}$, we consider now the space $L^2(\mathcal{I})$ of square integrable functions on $\mathcal{I}$ (i.e. $\int_{\mathcal{I}} f(s)^2 \text{d}s < \infty$). The inner product in $L^2(\mathcal{I})$ is defined for any $f \in L^2(\mathcal{I})$ and $g \in L^2(\mathcal{I})$ by $\langle f, g \rangle_{L^2} = \int_{\mathcal{I}} f(s)g(s) \text{d}s$. In the multivariate functional data setting (see \cite{Happ2018}), we can define the inner product between two multivariate functions $\bold{f} \in L^2(\mathcal{I}){}^p$ and $\bold{g} \in L^2(\mathcal{I}){}^p$ as
\begin{equation*}
    \langle \bold{f}, \bold{g} \rangle = \sum_{j=1}^p \langle f_j, g_j \rangle_{L^2} = \sum_{j=1}^p \int_{\mathcal{I}} f_j(s) g_j(s)\text{d} s
\end{equation*}
which defines a Hilbert Space. Considering now the space $\mathscr{H}=L^2(\mathcal{I}){}^{p_1 \times \dots \times p_D}$ of tensors of order $D$ with entries in $L^2(\mathcal{I})$. We define the inner product between two functional tensors $\mathcal{F} \in \mathscr{H}$ and $\mathcal{G} \in \mathscr{H}$ as 
\begin{equation}
    \label{eq:innerprod}
    \langle \mathcal{F}, \mathcal{G} \rangle = \sum_{\bold{j}} \langle f_{\bold{j}}, g_{\bold{j}} \rangle_{L^2} = \sum_{\bold{j}} \int_{\mathcal{I}} {f_{\bold{j}}}(s) {g_{\bold{j}}}(s) \text{d} s
\end{equation}
where we denote $f_{\bold{j}}(s) = (f(s))_{\bold{j}}$. Using this inner product, it can be shown easily that the space $\mathscr{H}$ is a Hilbert space.


\subsection{Existing methods}

The CANDECOMP/PARAFAC (CP) decomposition (\cite{Harshman1970}, \cite{Carroll1970}), also referred to as the Canonical Polyadic Decomposition (CPD), is a powerful and well-known method in the tensor literature. Its purpose is to represent any tensor as a sum of rank-$1$ tensors. In $\mathbb{R}^{p_1 \times \dots \times p_D}$, the rank-$R$ CPD is defined as

\begin{equation}
\label{eq:PARAFAC}
 \sum_{r=1}^{R} \mathbf{a}_{1r} \circ \dots \circ \mathbf{a}_{Dr} = \llbracket \bold{A}_1; \dots; \bold{A}_D \rrbracket
\end{equation}
where $\mathbf{a}_{dr} \in \mathbb{R}^{p_d}$ and $\mathbf{A}_d = [\mathbf{a}_{d1}, \dots, \mathbf{a}_{dR}] \in \mathbb{R}^{p_d \times R}$. Here, each vector $\mathbf{a}_{dr}$ captures latent information along the mode $d$ of the tensor. Consequently, the CPD can be effectively employed as a dimension reduction technique, particularly compelling in the tensor setting, as the number of variables grows exponentially with the order of the tensor.

To approximate any tensor $\mathcal{X} \in \mathbb{R}^{p_1 \times \dots \times p_D}$ using a CPD, factor matrices $\bold{A}_d$ are estimated by minimizing the squared difference between the original tensor $\mathcal{X}$ and its CPD approximation,
\begin{equation}
\label{eq:PARAFAC:optim}
\min_{\bold{A}_1,\dots, \bold{A}_D} \| \mathcal{X} - \llbracket \bold{A}_1; \dots; \bold{A}_D \rrbracket \|_F^2
\end{equation}
The original optimization process used an Alternating Least Squares (ALS) procedure (\cite{Carroll1970}, \cite{Harshman1970}), where, iteratively, each factor matrix is updated while keeping the others fixed until convergence is reached. The procedure belongs to the more general class of block relaxation algorithms, known to be locally and globally convergent under mild assumptions (\cite{Lange2013}).

When trying to decompose a tensor with functional entries, this method may seem inefficient as it fails to integrate the inherent continuous structure of functional spaces. Additionally, the CP model assumes a consistent structure on each mode across dimensions. Therefore, it cannot handle schemes where the functional entries are sampled irregularly, e.g., at different time points. To answer this last problem, several alternatives to the CP model were introduced, especially the PARAFAC2 model (\cite{Hars1972b}, \cite{Kiers1999}), which was widely employed but still does not integrate any smoothness structure into the reconstructed tensor.

In the functional setting, the Karhunen–Loève theorem is often used to retrieve a low-dimensional representation of a high (or infinite) dimensional random process. The theorem guarantees the existence of an orthonormal basis ${\Phi_k}$ and a sequence ${\xi_k}$ of independent random variables for any squared integrable random process $X$, such that:
\begin{equation}
\label{eq:FPCA}
    X(t) = \sum_k^\infty \xi_k \Phi_k(t)
\end{equation}
By truncating the above sum to a limited number of functions, we can recover a good enough representation of $X$. More recently, this decomposition was extended to multivariate functional data, allowing the representation of any multivariate random process as a sum of multivariate orthonormal functions and scores. By arranging the sequence of independent random variables ${\xi_k}$ in descending order of variance, this decomposition can be viewed as an extension of the well-known principal component analysis (PCA). Similarly to the PCA, estimating parameters in model (\ref{eq:FPCA}) relies on the eigen decomposition of the covariance operator of the random process. As evoked previously, numerous approaches have been proposed to achieve this. The most notable of which, \cite{Yao2005}, proposed a linear least square smoothing procedure with asymptotic guarantees, allowing estimation of the covariance surface in sparse and irregular data sampling settings. Furthermore, using simple assumptions on the joint distribution of principal components and residual errors, they introduce an efficient approach for estimating principal components.

When dealing with multivariate or tensor-structured functional data, it appears natural to extend this decomposition. As evoked in Section \ref{sec:introduction}, various approaches have been introduced to extend this decomposition to multivariate functional data. In the tensor-structured functional data setting, vectorizing the tensor appears as a straightforward way to proceed. However, similar to before, as the number of processes grows exponentially with the order of the tensor, this procedure would be computationally unfeasible for large tensors, as it would require estimating orthonormal functions for each functional feature. In addition, the vectorization of the tensor would lead to the loss of the tensor structure, and the resulting representation would still be high-dimensional, causing the approach to be inefficient.

\subsection{Proposed model}

We propose to extend the rank $R$ CPD introduced in (\ref{eq:PARAFAC}) to functional tensors in $\mathscr{H}$ by further multiplying each rank-1 sub-tensor $\bold{a}_{1r} \circ \dots \circ \bold{a}_{Dr}$ by a functional term $\phi_{r} \in L^2(\mathcal{I})$,
\begin{equation}
\label{eq:model:lowrankf}
     \sum_{r=1}^R \phi_{r}(t) (\bold{a}_{1r} \circ \dots \circ \bold{a}_{Dr}) = \llbracket \boldsymbol{\Phi}(t);\bold{A}_1; \dots; \bold{A}_D \rrbracket \qquad t \in \mathcal{I}.
\end{equation}
denoting $\boldsymbol{\Phi} = (\phi_{1}, \dots, \phi_{r}) \in L^2(\mathcal{I}){}^R$. We propose to name this decomposition the functional PARAFAC (F-PARAFAC) decomposition. 

As discussed earlier, it is common for a functional tensor to have a sample mode, i.e., a mode on which sub-tensors can be seen as samples drawn from an underlying tensor distribution. To account for this inherent randomness, \cite{Lock2018a} introduced a probabilistic modeling in the CPD by considering sample mode vectors to be sampled from a random distribution. Similarly, we propose to account for the inherent randomness of random functional tensors in $\mathscr{H}$ by multiplying each rank-$1$ functional tensor $\phi_{r}(t) (\bold{a}_{1r} \circ \dots \circ \bold{a}_{Dr})$ by a random variable $u_r$,
\begin{equation}
\label{eq:model:lowrankfu}
    \sum_{r=1}^R u_r \phi_{r}(t) (\bold{a}_{1r} \circ \dots \circ \bold{a}_{Dr}) = \llbracket \bold{U}; \boldsymbol{\Phi}(t);\bold{A}_1; \dots; \bold{A}_D \rrbracket \qquad t \in \mathcal{I}.
\end{equation}
denoting $\bold{U} = (u_1, \dots, u_R) \in \mathbb{R}^{1\times R}$. We propose naming this model the Latent Functional PARAFAC (LF-PARAFAC) decomposition. Assumptions on the distribution of $\bold{U}$ are further discussed in Section \ref{section:estimation}. Moreover, we adopt the following notations,  $\bold{A}_{(D)} = (\bold{A}_D \odot \dots \odot \bold{A}_{1})$, $\bold{A}_{(-d)} = (\bold{A}_D \odot \dots \odot \bold{A}_{d+1} \odot \bold{A}_{d-1} \odot \dots \odot \bold{A}_{1})$, and $\bold{A}_{(-d)}^{\boldsymbol{\Phi}}(t) = \bold{A}_{(-d)} \odot \boldsymbol{\Phi}(t)$. 

Given a random functional tensor $\mathcal{X}$, which we assume to be centered through this section, we are interested in finding the best LF-PARAFAC decomposition approximating $\mathcal{X}$:
\begin{equation}
    \label{eq:model:optimization}
    (\boldsymbol{\Psi}^*, \boldsymbol{\Phi}^*, \bold{A}_1^*, \dots, \bold{A}_D^*) = \argmin_{\boldsymbol{\Psi}, \boldsymbol{\Phi}, \bold{A}_1, \dots, \bold{A}_D}  \mathbb{E} [\| \mathcal{X} - \llbracket \boldsymbol{\Psi}(\mathcal{X}); \boldsymbol{\Phi}; \bold{A}_1; \dots; \bold{A}_D \rrbracket\|^2]
\end{equation}
In this optimization problem, $\boldsymbol{\Psi}^*$ is the operator $\mathcal{X} \mapsto \boldsymbol{\Psi}^*(\mathcal{X}) = \bold{U}^*$ outputting the optimal sample mode vector for any given random functional tensor $\mathcal{X}$. The asymmetry between $\bold{U}$, which is random, and $\boldsymbol{\Phi}, \bold{A}_1, \dots, \bold{A}_D$, which are fixed, allows us to rewrite the optimization problem as
\begin{equation}
    \label{eq:model:optimizationZ}
    \min_{\boldsymbol{\Phi}, \bold{A}_1, \dots, \bold{A}_D}  \mathbb{E}[\, \min_{\bold{U}}\| \mathcal{X} - \llbracket \bold{U}; \boldsymbol{\Phi}; \bold{A}_1; \dots; \bold{A}_D \rrbracket\|^2]
\end{equation}
Proposition \ref{proposition:psi} gives the expression of $\Psi^*$, corresponding to the optimal matrix $\bold{U}$ of the optimization sub-problem on the right-hand side of (\ref{eq:model:optimizationZ}). 
\begin{proposition}
    \label{proposition:psi}
    Fixing $\boldsymbol{\Phi}, \bold{A}_1, \dots, \bold{A}_D$, for any $\mathcal{X}$, the optimal operator $\boldsymbol{\Psi}^*$ is given by
    \begin{equation}
    \label{eq:solz}
        \boldsymbol{\Psi}^*(\mathcal{X}) = \left ( \int \bold{x}(t)^T (\bold{A}_{(D)} \odot \boldsymbol{\Phi}(t)) \text{d} t \right ) \left ( \int (\bold{A}_{(D)} \odot \boldsymbol{\Phi}(t))^T (\bold{A}_{(D)} \odot \boldsymbol{\Phi}(t)) \text{d} t \right )^{-1}
    \end{equation}
\end{proposition}
Plugging the optimal operator $\Psi^*$ in model (\ref{eq:model:lowrankfu}) eliminates the dependency on $\bold{U}$, and the problem can now be seen as depending only on $\boldsymbol{\Phi}, \bold{A}_1, \dots, \bold{A}_D$:
\begin{equation}
    \label{eq:model:optimization2}
    (\boldsymbol{\Phi}^*, \bold{A}_1^*, \dots, \bold{A}_D^*) = \argmin_{\boldsymbol{\Phi}, \bold{A}_1, \dots, \bold{A}_D}  \mathbb{E} [\| \mathcal{X} - \llbracket \bold{U}^*; \boldsymbol{\Phi}; \bold{A}_1; \dots; \bold{A}_D \rrbracket\|^2]
\end{equation}
Denoting collapsed functional covariance matrices $\boldsymbol{\Sigma}_{[f]}(s, t) = \mathbb{E} [\bold{x}(s) (\bold{I}_{p_{(d)}} \otimes \bold{x}(t)^T)]$ and $\boldsymbol{\Sigma}_{[d]}(s, t) = \mathbb{E} [\bold{X}_{(d)}(s) (\bold{I}_{p_{(-d)}} \otimes \bold{x}(t)^T)]$, for which an interpretation of $\boldsymbol{\Sigma}_{[f]}(s, t)$ and $\boldsymbol{\Sigma}_{[d]}(s, t)$ is given in section \ref{subsec:covarianceestimation}, we can derive Corollary \ref{corollary:cov} from Proposition \ref{proposition:psi}.
\begin{corollary}
\label{corollary:cov}
Denoting $\boldsymbol{\Lambda}^* = \mathbb{E}[{\bold{U}^*}^T \bold{U}^*]$ the covariance matrix of $\bold{U}^*$ (note that $\boldsymbol{\Lambda}^* \in \mathbb{R}^{R\times R}$ since $\bold{U}^* \in \mathbb{R}^{1\times R}$), and $\bold{K}(t) =  (\bold{A}_{(D)} \odot \boldsymbol{\Phi}(t)) \left ( \int (\bold{A}_{(D)} \odot \boldsymbol{\Phi}(t))^T (\bold{A}_{(D)} \odot \boldsymbol{\Phi}(t)) \right )^{-1}$, we have
\begin{equation}
        \boldsymbol{\Lambda}^* = \int \int \bold{K}(s)^\top \boldsymbol{\Sigma}_{[f]}(s, t) \bold{K}(t) \text{d} s \text{d} t
    \end{equation}
\end{corollary}

\begin{lemma}
\label{lemma:1}
    The objective function of optimization problem (\ref{eq:model:optimization2}) only depends on covariance $\boldsymbol{\Lambda}^*$ and other parameters via an operator $C$:
    \begin{equation*}
    \label{eq:optimizationC}
        (\boldsymbol{\Phi}^*, \bold{A}_1^*, \dots, \bold{A}_D^*) = \argmin_{\boldsymbol{\Phi}, \bold{A}_1, \dots, \bold{A}_D} C(\boldsymbol{\Lambda}^*, \boldsymbol{\Phi}, \bold{A}_1, \dots, \bold{A}_D)
    \end{equation*}
\end{lemma}

For a given tensor, the operator $C$ depends for a given tensor $\mathcal{X}$ on the collapsed covariance $\boldsymbol{\Sigma}_{[f]}$. This dependency can be replaced by a dependency on $\boldsymbol{\Sigma}_{[d]}$ for any $d \in [D]$, which gives multiple expression for $C$, introduced in Appendix \ref{sec:annex:proofs}. Expression are given in Appendix \ref{sec:annex:exp}. Lemma \ref{lemma:1} allows to derive simple solutions to (\ref{eq:optimizationC}), which are introduced in Proposition \ref{proposition:solutions}.

\begin{proposition}
\label{proposition:solutions}
    For the functional mode, we have
    \begin{equation}
        \label{eq:solf}
        \boldsymbol{\Phi}^*(s) =  \left ( \int \boldsymbol{\Sigma}_{[f]}(s, t) (\bold{A}_{(D)} \odot \bold{K}(t)) \text{d} t \right ) \left ( (\bold{A}_{(D)}^T \bold{A}_{(D)}) \ast \boldsymbol{\Lambda}^* \right )^{-1}
    \end{equation}
    For any $d\in [D]$ tabular mode we have,
    \begin{equation}
    \begin{split}
        \label{eq:sold}
        \bold{A}_d^* = \left ( \int \int \boldsymbol{\Sigma}_{[d]}(s, t) (\bold{A}_{(-d)}^{\boldsymbol{\Phi}} (s) \odot \bold{K}(t)) \text{d} s  \text{d} t  \right ) \left ( \int \bold{A}_{(-d)}^{\boldsymbol{\Phi}}(t)^T \bold{A}_{(-d)}^{\boldsymbol{\Phi}} (t) \text{d} t \ast \boldsymbol{\Lambda}^* \right )^{-1}
    \end{split}
    \end{equation}
\end{proposition}

\subsection{Identifiability}
\label{subsec:identifiability}

A well-known advantage of the native CPD is the identifiability of model parameters under mild assumptions. In the probabilistic setting, we define identifiability as the fact that if for two sets of model parameters $\boldsymbol{\Theta} = (\boldsymbol{\Lambda}, \boldsymbol{\Phi}, \bold{A}_1, \dots, \bold{A}_D)$ and $\tilde{\boldsymbol{\Theta}} = (\tilde{\boldsymbol{\Lambda}}, \tilde{\bold{A}}_f, \tilde{\bold{A}}_1, \dots, \tilde{\bold{A}}_D)$ we have, $L( \boldsymbol{\Theta}| \mathcal{X} ) = L( \tilde{\boldsymbol{\Theta}} | \mathcal{X} )$, for any data $\mathcal{X}$, where $L( \boldsymbol{\Theta}| \mathcal{X} )$ stands for the model likelihood, then $\boldsymbol{\Theta} = \tilde{\boldsymbol{\Theta}}$ (\cite{Lock2018a}). In this context, two well-known indeterminacies, preventing parameters from being identifiable, can be noted. The first is the scaling indeterminacy, which stems from the fact that $\llbracket \bold{U}; \boldsymbol{\Phi}; \bold{A}_1; \dots; \bold{A}_D \rrbracket = \llbracket \tilde{\bold{U}}; \tilde{\bold{A}}_f; \tilde{\bold{A}}_1; \dots; \tilde{\bold{A}}_D \rrbracket$ for any $\{ \tilde{\bold{U}}, \tilde{\bold{A}}_f, \tilde{\bold{A}}_1, \dots, \tilde{\bold{A}}_D \}= \{ a_{\bold{U}}\bold{U}, a_{f}\boldsymbol{\Phi}, a_{1}\bold{A}_1, \dots, a_{D}\bold{A}_D\}$ with $a_{\bold{U}} a_f a_1 \dots a_D = 1$. The second is the permutation/order indeterminacy, which comes from the fact that any reordering of the columns of the parameters matrices gives the same CPD tensor. To remove these two indeterminacies, we consider the following assumptions:
\begin{enumerate}[label=({C\arabic*})]
    \item \label{asp:C1} For any $d\in[D]$ and $r\in[R]$, $\| \bold{a}_{dr} \|_2 = 1$, and the first non-zero value of $\bold{a}_{dr}$ is positive. Similarly $\| \phi_{r} \|_{\mathcal{H}} = 1$ and $\phi_{r}$ is positive at the lower bound of its support.
    \item \label{asp:C2} Function $\phi_{r}$, and vectors $\bold{a}_{dr}$ are ordered such that sample-mode random variables $u_r$ have decreasing variance: $\boldsymbol{\Lambda}_{rr} < \boldsymbol{\Lambda}_{r'r'}$ for any $r < r'$ (we suppose no equal variance).
\end{enumerate}

Considering the scaling and order indeterminacies fixed, a simple condition on the k-rank of the feature matrices and $R$ was derived (\cite{Sidiropoulos2000}) for tensors of any order to guarantee identifiability. In the probabilistic setting, \cite{Lock2018a}, introduced a similar condition, which can be easily extended to our functional setting. Denoting $\krank(\bold{A}_d)$ the k-rank of the feature matrix $\bold{A}_d$, the maximum number $k$ such that any $k$ columns of $\bold{A}_d$ are linearly independent, Proposition \ref{proposition:identifiability} gives conditions for which model (\ref{eq:model:lowrankfu}) is identifiable.

\begin{proposition}
\label{proposition:identifiability}
Denoting $\mathscr{H}$ the space to which belong functions $\phi_r$, considering assumptions \ref{asp:C1} and  \ref{asp:C2} verified, model (\ref{eq:model:lowrankfu}) is identifiable if $\mathscr{H}$ is infinite dimensional and if the following inequality is verified:
\begin{equation}
    \sum_{d=1}^D \krank(\bold{A}_d) \geq D - 1
\end{equation}
Conversely assuming $\mathscr{H}$ is finite dimensional and denoting $\bold{C}$ any basis decomposition matrix for functions $\phi_r$, model (\ref{eq:model:lowrankfu}) is identifiable if deleting any row of $\bold{C} \odot \bold{A}_1 \odot \dots \odot \bold{A}_D$ results in a matrix having two distinct full-rank submatrices and if we have:
\begin{equation}
\krank(\bold{C}) + \sum_{d=1}^D \krank(\bold{A}_d) \geq R + (D - 1)
\end{equation}
\end{proposition}

\section{Estimation}
\label{section:estimation}

\subsection{Observation model}
\label{subsec:observationmodel}

We now consider $n$ samples of a random functional tensor $\mathcal{X}$ of order $D$, denoted $\{\mathcal{X}_i\}_{1 \leq i \leq n}$, with an underlying random functional CPD structure. For a given $\mathcal{X}_i$, we assume that entries of $\mathcal{X}_i$ are observed at similar time points. Observation time points of $\mathcal{X}_i$ are denoted $\bold{t}_i = (t_{i1}, \dots, t_{iN_i})$. The number of observations is allowed to be different across samples $\mathcal{X}_i$. In almost all settings, observations are contaminated with noise. Denoting $\mathcal{E}$ the random functional tensor modeling this noise, and $\mathcal{E}_{ik}$ the measurement error associated with $\mathcal{X}_{i}(t_{ik})$, we consider that $\mathcal{E}_{ik}$ has i.i.d. entries drawn from a distribution $\mathcal{N}(0, \sigma^2)$. The $k$th tensor-structured observation of the sample $i$, denoted $\mathcal{Y}_{ik}$, is thus modeled as, 
\begin{equation}
    \label{eq:obsmodel}
    \mathcal{Y}_{ik} = \mathcal{X}_{i}(t_{ik}) + \mathcal{E}_{ik}
\end{equation}
The noise term can also be seen as modeling the unexplained part in the CPD. In this context, our goal is to find the parameters of the random functional CPD which approximates the best observations of the tensor $\mathcal{X}$.

\subsection{Covariance terms}
\label{subsec:covarianceestimation}

In (\ref{eq:solf}) and (\ref{eq:sold}), we must estimate terms $\boldsymbol{\Sigma}_{[f]}(s, t) = \mathbb{E} [\bold{x}(s) (\bold{I}_{P} \otimes \bold{x}(t)^T)]$ and $\boldsymbol{\Sigma}_{[d]}(s, t) = \mathbb{E} [\bold{X}_{(d)}(s) (\bold{I}_{p_{(-d)}} \otimes \bold{x}(t)^T)]$ to retrieve optimal functions $\boldsymbol{\Phi}^*$ and factor matrices $\bold{A}_d^*$. As evoke previously, for any $d \in [D]$, the term $\boldsymbol{\Sigma}_{[d]}(s, t)$ can be seen as a transformation of the mode-$d$ functional covariance matrix $\boldsymbol{\Sigma}_{(d)}(s, t) = \mathbb{E}[\bold{x}_{(d)}(s) \bold{x}_{(d)}(t)^\top] \in \mathbb{R}^{P \times P}$. Recalling that $P = \prod_{d'} p_{d'}$ and $p_{(-d)} = \prod_{d'\neq d} p_{d'}$, the transformation consists of a row-space collapse of non-mode-$d$ entries into the column space, resulting in a matrix of dimensions $p_d\times P \cdot p_{(-d)}$. 

To better grasp the idea, let's consider an order-$2$ functional tensor of dimensions $p_1\times p_2$. The mode-1 functional covariance matrix can be decomposed as,
\begin{equation*}
\boldsymbol{\Sigma}_{(1)}(s,t) =
    \begin{bmatrix}
        \boldsymbol{\Sigma}_{(1, 1)}(s,t) & \dots  & \boldsymbol{\Sigma}_{(1, p_2)}(s,t)\\
        \vdots & \ddots  & \vdots \\
        \boldsymbol{\Sigma}_{(p_2, 1)}(s,t) & \dots  & \boldsymbol{\Sigma}_{(p_2, p_2)}(s,t)\\
    \end{bmatrix} \in \mathbb{R}^{p_1p_2 \times p_1p_2}
\end{equation*}
where each block functional matrix $\boldsymbol{\Sigma}_{(j, j')}(s, t) = \mathbb{E}[\mathcal{X}_{.j}(s) \mathcal{X}_{. j'}(t)^\top] \in \mathbb{R}^{p_1\times p_1}$ corresponds to functional covariance of mode-$1$ variables with fixed mode-$2$ variables $j, j' \in [p_2]$. The functional matrix $\boldsymbol{\Sigma}_{[1]}(s, t)$ can then be expressed as,
\begin{equation*}
\boldsymbol{\Sigma}_{[1]}(s,t) =
    \begin{bmatrix}
        [\boldsymbol{\Sigma}_{(1, 1)}(s,t) & \dots & \boldsymbol{\Sigma}_{(1, p_2)}(s,t)] \; \dots \; [\boldsymbol{\Sigma}_{(p_2, 1)}(s,t) & \dots & \boldsymbol{\Sigma}_{(p_2, p_2)}(s,t)]
    \end{bmatrix} \in \mathbb{R}^{p_1 \times p_1p_2p_2}
\end{equation*}
For $\boldsymbol{\Sigma}_{[2]}(s, t) \in  \mathbb{R}^{p_2\times p_1 p_2 p_1}$ a similar reasoning can help see the structure of the matrix with respect to the functional block entries of $\boldsymbol{\Sigma}_{(2)}$.

Since functional matrices $\boldsymbol{\Sigma}_{(d)}$ can be obtained by reordering $\boldsymbol{\Sigma}$, only $\boldsymbol{\Sigma}$ must be estimated to obtain an estimate of $\boldsymbol{\Sigma}_{[d]}$. For this purpose, we rely on a local linear smoothing procedure outlined in \cite{Fan2018}. First, assuming functional tensor entries are not centered ($\mathcal{M} \neq 0$), a centering pre-processing step is carried out. The mean function $m_{\bold{j}} = \mathbb{E}[y_{\bold{j}}]$, which is estimated by aggregating observations and using the smoothing procedure evoked earlier, is removed from each functional entry of $\mathcal{Y}$. Next, for each pair of \textit{centered} functional entries $\mathcal{Y}_{\bold{j}}$ and $\mathcal{Y}_{\bold{j}'}$, we estimate the (cross-)covariance surface $\Sigma_{\bold{j}\bold{j}'}(s, t)$ using a similar approach. The raw-(cross-)covariance is computed for each pair of observations within each sample. Raw (cross-)covariances are then aggregated, and the values are smoothed on a regular 2-dimensional grid using the same smoothing procedure as before.

Additionally, we can estimate $\sigma$ by averaging the residual variance of functional entry. This residual variance is obtained for each process by running a one-dimensional smoothing of the diagonal using the previously removed pairs of observations and then averaging the difference between this smoothed diagonal and the previously computed smoothed diagonal (without the noise). For more details on this procedure, we refer the reader to the seminal paper by \cite{Yao2005}.

This traditional smoothing approach was chosen for its simplicity and theoretical guarantees (\cite{Yao2005}). However, several other approaches exist to estimate (cross-)covariance surfaces. Most notably, different smoothing procedures can be used (\cite{Eilers2003}, \cite{Wood2003}). For a univariate process, the FAst Covariance Estimation (FACE) proposed by \cite{Xiao2014} allows estimating the covariance surface with linear complexity (with respect to the number of observations). An alternative FACE method, robust to sparsely observed functional data, was introduced in \cite{Xiao2017}. More recently, the method was adapted to multivariate functional data in \cite{Li2020}.


\subsection{Solving procedure}
\label{subsection:solvingprocedure}

Solution equations introduced in Proposition \ref{proposition:solutions} along the expression of sample mode optimal vector covariance matrix introduced in Corollary \ref{corollary:cov} emphasize the use of a block relaxation algorithm outlined in Algorithm \ref{algo:1} to estimate parameters in (\ref{eq:model:lowrankfu}).
\begin{algorithm}
\caption{Block relaxation procedure for solving (\ref{eq:model:lowrankfu})}\label{algo:1}
\KwResult{$\hat{\boldsymbol{\Lambda}}, \hat{\boldsymbol{\Phi}}, \hat{\bold{A}_1}, \dots, \hat{\bold{A}_D}$}
Initialize $\boldsymbol{\Lambda}^{(0)}, \boldsymbol{\Phi}^{(0)}, \bold{A}_1^{(0)}, \dots, \bold{A}_D^{(0)}$\\
\Repeat{$C(\boldsymbol{\Lambda}^{(i+1)}, \boldsymbol{\Phi}^{(i+1)}, \bold{A}_1^{(i+1)},\dots, \bold{A}_D^{(i+1)}) - C(\boldsymbol{\Lambda}^{(i)}, \boldsymbol{\Phi}^{(i)}, \bold{A}_1^{(i)},\dots, \bold{A}_D^{(i)}) < \varepsilon$}{
    \begin{align*}
    &\boldsymbol{\Lambda}^{(i+1)} \leftarrow \int \int \bold{K}(s)^\top \hat{\boldsymbol{\Sigma}}(s, t) \bold{K}(t) \text{d} s \text{d} t \\
    &\boldsymbol{\Phi}^{(i+1)} \leftarrow \argmin_{\boldsymbol{\Phi}} C(\boldsymbol{\Lambda}^{(i+1)}, \boldsymbol{\Phi}, \bold{A}_1^{(i)},\dots, \bold{A}_D^{(i)})
    \end{align*}
\For{$d=1, \dots, D$}{
\begin{equation*}
\begin{split}
\bold{A}_d^{(i+1)} \leftarrow \argmin_{\bold{A}_d} C(\boldsymbol{\Lambda}^{(i+1)}, \boldsymbol{\Phi}^{(i+1)}, \bold{A}_1^{(i+1)}, \dots, \bold{A}_{d-1}^{(i+1)}, \bold{A}_d, \bold{A}_{d+1}^{(i)}, \dots, \bold{A}_D^{(i)})
\end{split}
\end{equation*}
}
}
\end{algorithm}

This algorithm uses all non-mode-$d$ collapsed functional covariance matrices $\boldsymbol{\Sigma}_{[d]}$ as input. It produces estimates $\hat{\boldsymbol{\Lambda}}, \hat{\boldsymbol{\Phi}}, \hat{\bold{A}_1}, \dots, \hat{\bold{A}_D}$ of the parameters in model (\ref{eq:model:lowrankfu}). 

Initialization can significantly influence the algorithm's output, particularly in high-dimensional settings. In practice, random initialization is often chosen for simplicity, but it comes with the cost of running the algorithm multiple times. In our setting, we suggest initializing feature functions and matrices with the ones obtained from a standard CPD. Experience shows that this simple approach is stable and provides good results. Note that alternative approaches were proposed in \cite{Guan2023} and \cite{Han2023}, but were not studied here.

\subsection{Sample-mode inference}

For a given sample $\boldsymbol{\mathcal{Y}}_i$, when dealing with sparse and irregular observations, the estimation of $\bold{U}_i$ given by (\ref{eq:solz}) can be unstable (or untracktable), since it requires integrating $\bold{x}_i$ on possibly very few observation time points. Furthermore, this estimator can give inconsistent estimates as it does not incorporate the error term $\mathcal{E}$. Inspired from \cite{Yao2005}, we propose to get around this problem by assuming a jointly Gaussian distribution on vectors $u_{ir}$ and residual errors $\mathcal{E}_{i,k}$. In this context, we thus consider $\bold{U} \sim \mathcal{N}(\boldsymbol{\mu}, \boldsymbol{\Lambda})$. As done in \cite{FGCCA} we propose to generalize the conditional expectation estimator of the PACE-FPCA to functional tensors. Proposition \ref{proposition:scores} gives an expression of the conditional expectation of $\bold{u}_i$

\begin{proposition}
\label{proposition:scores}
    Denoting $\bold{U}_i = (u_{i1}, \dots, u_{iR}) \in  \mathbb{R}^{1 \times R}$ $\bold{A}_{f,i} = [ \boldsymbol{\Phi}(t_{i1})^\top, \dots, \boldsymbol{\Phi}(t_{iN_i})^\top ]^\top \in \mathbb{R}^{N_i \times R}$, $\bold{F}_i = \bold{A}_{(D)} \odot \bold{A}_{f,i} \in \mathbb{R}^{N_i P \times R}$, assuming joint Gaussian distribution on $\bold{U}$ and residual errors, a Bayesian estimator of $\bold{U}_i$ is 
    \begin{equation}
        \label{eq:bayes_estimator}
        \tilde{\bold{U}}_i = \mathbb{E}[\bold{U}_i | \boldsymbol{\mathcal{Y}}_i] = \boldsymbol{\Lambda} \bold{F}_i^\top (\bold{F}_i \boldsymbol{\Lambda} \bold{F}_i^\top + \sigma^2\bold{I})^{-1} [\boldsymbol{\bold{y}}_i - \boldsymbol{\mu}_i]
    \end{equation}
\end{proposition}

Coupled with estimations $\hat{\boldsymbol{\Lambda}}, \hat{\boldsymbol{\Phi}}, \hat{\bold{A}_1}, \dots, \hat{\bold{A}_D}$ obtained from Algorithm \ref{algo:1}, and estimations $\hat{\sigma}$ and $\hat{\boldsymbol{\mu}}_i$ of $\sigma$ and $\boldsymbol{\mu}_i$, obtained as described in Section \ref{subsec:covarianceestimation}, Proposition \ref{proposition:scores} allows to derive the following empirical Bayes estimator for $\bold{U}_i$:
\begin{equation}
        \hat{\bold{U}}_i = \hat{\boldsymbol{\Lambda}} \hat{\bold{F}}_i^\top (\hat{\bold{F}}_i \hat{\boldsymbol{\Lambda}} \hat{\bold{F}}_i^\top + \hat{\boldsymbol{\Sigma}}_i)^{-1} [\boldsymbol{\bold{y}}_i - \hat{\bold{m}}_i]
    \end{equation}

\subsection{Rank selection}

We suggest using two methods for selecting the rank of the decomposition. The first uses a likelihood cross-validation criterion as introduced in \cite{Lock2018a}. In this setting, the data is split between a training set $\mathcal{Y}^{\text{train}}$ and a test set $\mathcal{Y}^{\text{test}}$. The rank is then chosen such that it maximizes the log-likelihood of $\mathcal{Y}^{\text{test}}$ using parameters $\boldsymbol{\Theta}^{\text{train}}$ estimated with $\mathcal{Y}^{\text{train}}$. Denoting the covariance of $\bold{y}_i$, the vectorization of $\mathcal{Y}_i$, $\boldsymbol{\Sigma}_{\bold{y}, i} = (\bold{A}_{D} \odot \bold{A}_{f,i}) \boldsymbol{\Lambda} (\bold{A}_{D} \odot \bold{A}_{f,i})^\top + \hat{\sigma}^2 \bold{I}_{N_iP}$, $\boldsymbol{\Theta} = (\boldsymbol{\Lambda}, \boldsymbol{\Phi}, \bold{A}_1, \dots, \bold{A}_D)$ the parameters, $\bold{y}_i$ the vector of observations of sample $i$, and $\bold{m}_{i}$ the vectorized mean tensor of $\mathcal{Y}$, at time points $\bold{t}_i$, we can derive the log-likelihood of (\ref{eq:obsmodel}):
\begin{equation}
    L(\boldsymbol{\Theta}| \boldsymbol{\mathcal{Y}}) = -\frac{1}{2} \sum_{i=1}^n (\bold{y}_{i} - \bold{m}_{i})^\top \boldsymbol{\Sigma}_{\bold{y},i}^{-1} (\bold{y}_{i} - \bold{m}_{i}) - \frac{1}{2} \log{|\boldsymbol{\Sigma}_{\bold{y}, i}|}
\end{equation}
Using the expression of the likelihood, we thus aim at maximizing:
\begin{equation}
    \text{LCV}(R) = \frac{1}{K} \sum_{k=1}^K L(\boldsymbol{\Theta}^{\text{train}}_k | \boldsymbol{\mathcal{Y}}^{\text{test}}_k)
\end{equation}
where the index $k$ denotes the fold number in the cross-validation procedure. This approach gives an accurate estimation of the rank but it can be computationally expensive if the number of folds $K$ is too large, notably in a leave-one-out context. The second method is based on the Akaike Information Criterion (AIC), and offers a cheaper and faster alternative. In this setting, the rank $R$ is chosen such that it minimizes:

 \begin{equation}
     \text{AIC}(R) = R - L(\hat{\boldsymbol{\Theta}}_R| \boldsymbol{\mathcal{Y}})
 \end{equation}

\section{Application study}
\label{section:application}

To help characterize cognitive decline among Alzheimer's Disease (AD) and Cognitively Normal (CN) patients, we propose to apply the introduced method to various neurocognitive markers measured over several years in the Alzheimer’s Disease Neuroimaging Initiative (ADNI) study. In this context, we consider 6 different cognitive markers observed over up to 10 years on 518 healthy controls with normal cognition and 370 patients diagnosed with Alzheimer's disease, resulting in $n=888$ total patients. Neurocognitive tests are often used by physicians to help assess the cognitive abilities of patients. In the context of Alzheimer's disease and, more generally, dementia, they play an important role in the diagnosis and follow-up. Most neurocognitive tests consist of more or less simple questions asked to patients regarding their everyday lives or to complete a task involving specific cognitive abilities. The 6 cognitive markers considered here are the Alzheimer's Disease Assessment Scale (ADAS), the Mini-Mental State Examination (MMSE), the Rey Auditory Verbal Learning Test (RAVLT), the Functional Activities Questionnaire (FAQ), the Montreal Cognitive Assessment (MOCA), and the logical memory delayed recall score (LDEL). Scores were normalized for improved readability. Markers were observed at different time points across individuals, resulting in a heavily irregular sampling scheme with $1064$ unique time points overall. Only subjects with at least 2 observations were considered, resulting in an average number of observations per subject of 4.9. Trajectories are represented in Figure \ref{fig:traj}.

\begin{figure}
    \centering
    \includegraphics[scale=0.5]{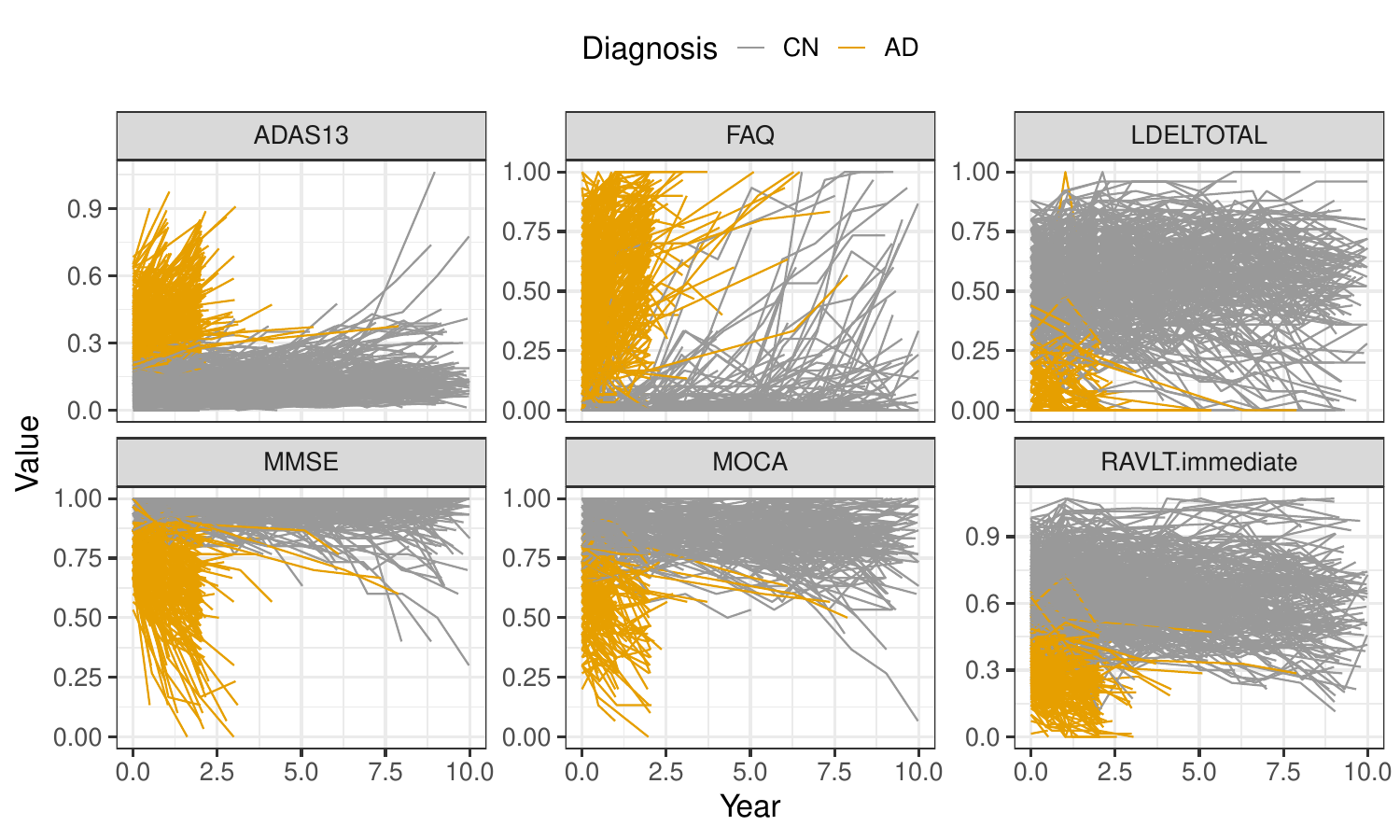}
    \caption{Trajectories of six cognitive scores measured over 10 years. Colors correspond to baseline diagnosis: Cognitively Normal (CN), Alzheimer's Disease (AD)}
    \label{fig:traj}
\end{figure}


The LF-PARAFAC decomposition was applied to the resulting random functional "tensor", which lies in $L^2(\mathcal{I})^6$, with $\mathcal{I} = [0,10]$ and can be seen as a tensor with dimensions $\text{subject} \times \text{time} \times \text{marker}$. The rank $R=4$ of the decomposition was selected so that results remained the same when using different initializations. The smoothing procedure' bandwidths were manually set to $1$. The standard PARAFAC decomposition was tested on the traditional tensor of dimensions $880 \times 1068 \times 6$ but failed to provide results due to the large dimension of the time mode and the high number of missing values. The 4 functions and vectors retrieved from the decomposition are represented in Figure \ref{fig:loadings} along with population mode scores.

\begin{figure}
    \centering
    \includegraphics[scale=0.44, align=t]{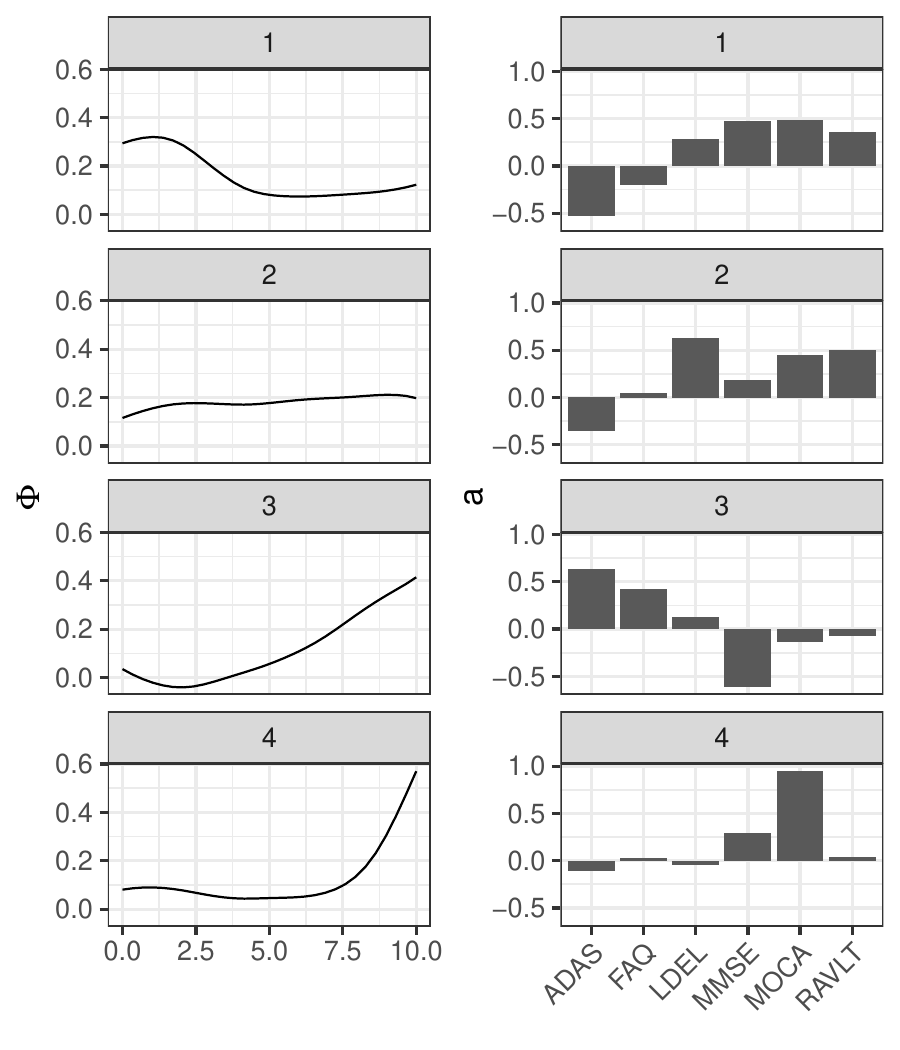}
    \includegraphics[scale=0.46, align=t]{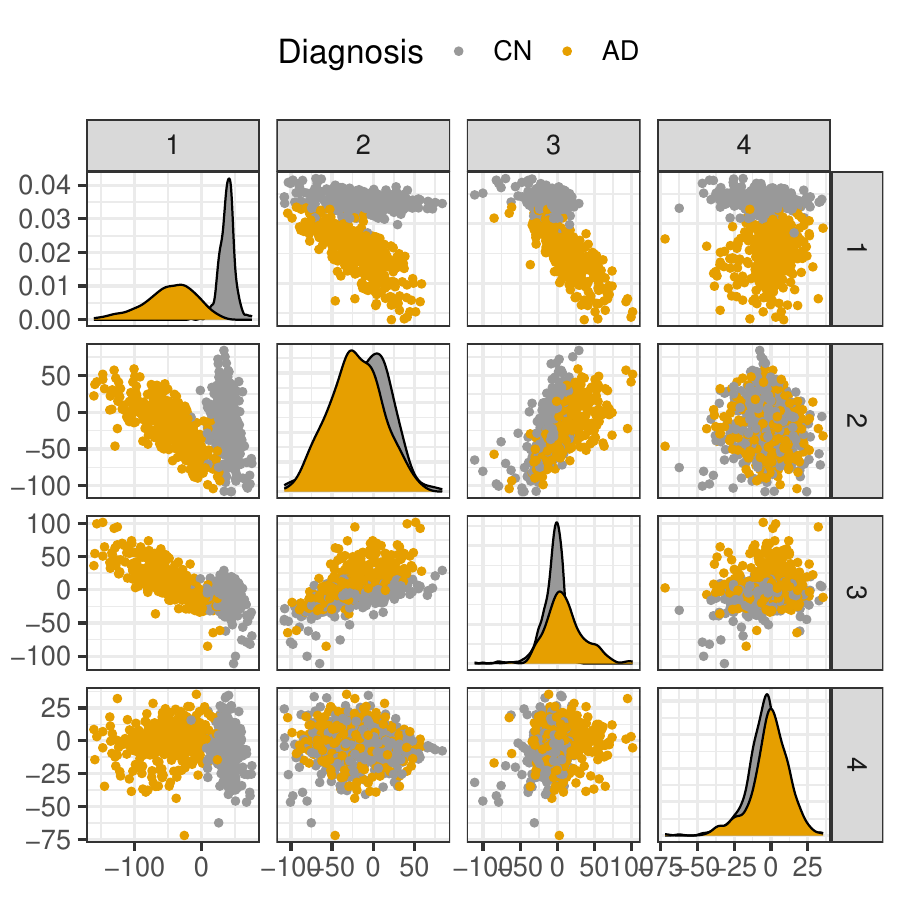}
    \caption{(left) Functions and vectors retrieved by a rank $R=4$ LF-PARAFAC decomposition. (right) Population mode scores colored by diagnosis at baseline: Cognitively Normal (CN), Alzheimer's Disease (AD)}
    \label{fig:loadings}
\end{figure}

Figures show that the LF-PARAFAC decomposition successfully captured inherent information of the random functional tensor. Notably, as expected, the method has retrieved smooth continuous vectors for the functional mode, describing predominant tendencies among cognitive scores. When observed with feature mode loadings, it seems that the first function, which has the highest associated sample-mode score variance, captures a sharp deterioration of cognitive abilities, mostly over the first 5 years: decreasing LDEL, MMSE, MOCA, and RAVLT, along with increasing ADAS and FAQ. This trend is also asserted by the distribution of the associated sample-mode scores: AD patients, who typically show a notable deterioration in cognitive abilities, have score values well separated from CN patients' scores. The second function retrieved describes a stable trend of cognitive measures over the ten years. The associated vector tells us that stable and high MMSE, MOCA, and RAVLT values are associated with stable and almost null FAQ and stable and low ADAS. The third function, which is still associated with a high variance sample-mode score, describes a slowly increasing trend for ADAS and FAQ, along with a slowly decreasing trend for MMSE, MOCA, and RAVLT. This trend might describe a slow decline in cognitive abilities, not as sharp as the first function described earlier. Finally, the fourth function is hard to interpret in itself. The feature-mode vector suggests that it describes a trend shared by both MOCA and MMSE scores. We may associate this observation with the fact that MOCA and MMSE are closely related cognitive tests relying on similar questions and that the function captures this shared information.

\section{Simulations study}
\label{section:simulations}

Using model (\ref{eq:obsmodel}), we generate $N = 100$ functional tensors of rank $R \in \{3, 5, 10\}$ on the interval $\mathcal{I} = [0, 1]$. Functions $\phi_{r}$ are randomly generated using a Fourier basis, with $M=5$ basis functions, and elements of factor matrices $\bold{A}_d$ are sampled from a uniform distribution. As assumed previously, sample mode vectors $\bold{u}_i$ are drawn from a distribution $\mathcal{N}(\bold{0}_R, \boldsymbol{\Lambda})$, where $\boldsymbol{\Lambda}$ is the diagonal matrix $\diag(\{ r^2 \})_{R \geq r \geq 1}$. The functional tensor is sampled on a regular grid of $\mathcal{I}$ with size $K=30$. The resulting array is then sparsified by removing a proportion $s \in \{0.0, 0.2, 0.5, 0.8\}$ of observations. We consider two settings for tensor dimensions: one of order $D=2$, with $p_1 = 10$, and one of order $D=3$ with $p_1 = p_2 = 5$ (in Appendix). Observation tensors are contaminated with i.i.d. measurement errors of variance $\sigma^2 = 1$, and we adjust the signal-to-noise ratio (SNR) by multiplying original tensors by a constant $c_{\text{SNR}}$. Three settings are considered: $\text{SNR} \in \{0.5, 1, 2\}$.

We compare our approach to the native CPD (without smoothing and probabilistic modeling), the Functional Tensor Singular Value Decomposition (FTSVD) (\cite{Han2023}), and the Multivariate Functional Principal Component Analysis (MFPCA) (\cite{Happ2018}). Two comparison tasks are considered: a reconstruction task and a parameter recovery task. Since the FTSVD does not handle missing values, it is not considered when $s \neq 0$. MFPCA is only considered for the reconstruction task as it is not based on the functional CPD model. Furthermore, since the method only deals with multivariate functional data, we apply the mode on the vectorized tensor. To obtain an approximation of the original tensor, we then tensorize back the approximated functional vector. For the reconstruction task, the metric used to compare methods is the root mean squared error (RMSE) between original tensors $\mathcal{X}_i$ and approximations $\hat{\mathcal{X}}_i$. For the parameter recovery task, we compare the true functions $\phi_{r}$ and estimated functions $\hat{a}_{fr}$ using the maximum principal angle (\cite{Bjorck1973}). MFPCA is carried out using the $\texttt{R}$ package $\texttt{MFPCA}$, the CPD is obtained using the \texttt{multiway} package, and the FTSVD is performed using the code from the associated paper. Simulations were run on a standard laptop.

\begin{figure}
    \centering
    \includegraphics[width=0.95\textwidth]{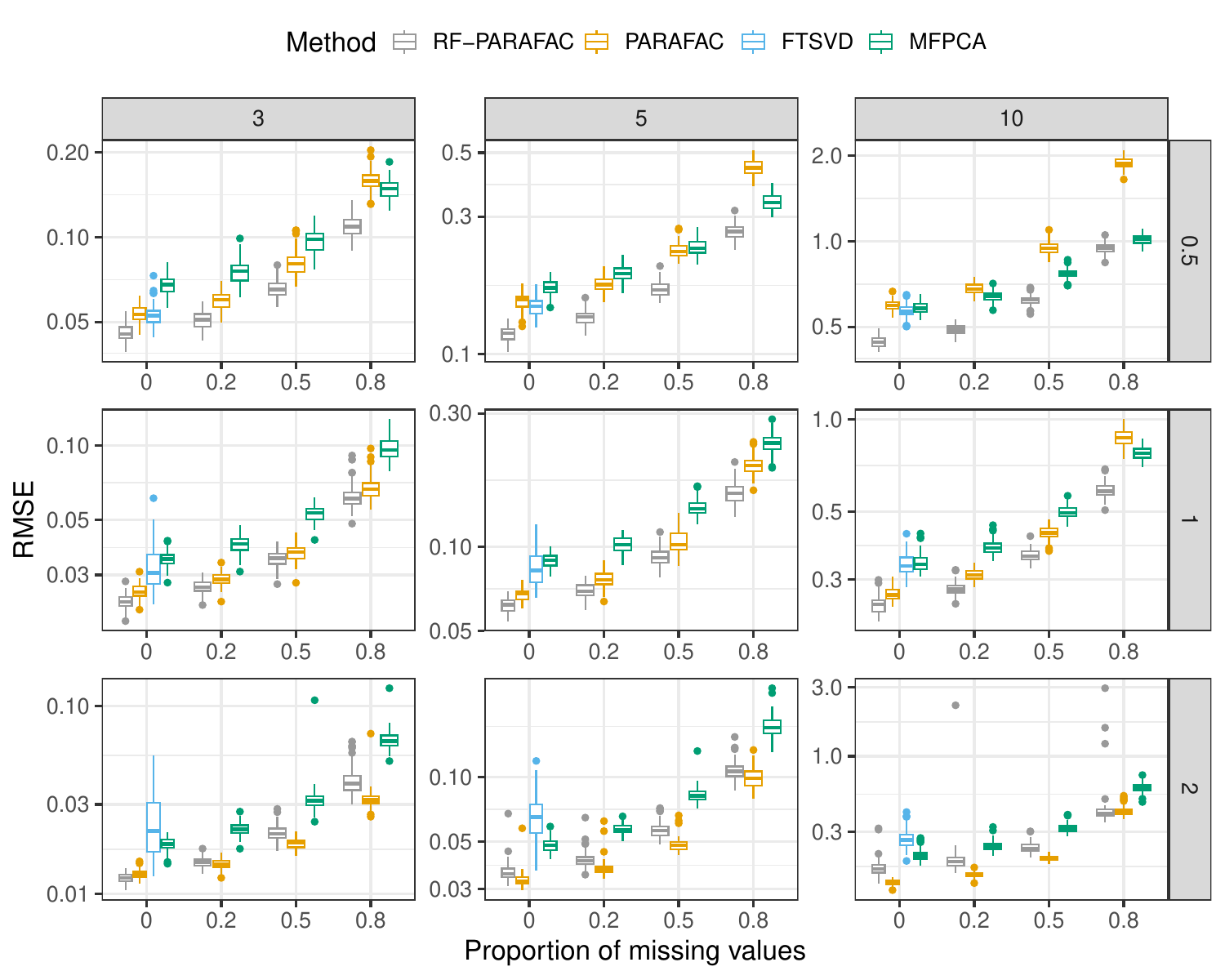}
    \caption{Reconstruction error on order $D=2$ tensors using $p_1 = 10$, with $N=100$ samples. Comparing LF-PARAFAC, PARAFAC (standard), FTSVD, and MFPCA for different values of $R$ (column-wise facets), different proportions of missing values (x-axis), and different signal-to-noise ratios (SNR) (row-wise facets). Plot obtained from $100$ simulation runs.}
    \label{fig:sim}
\end{figure}

\begin{figure}
    \centering
    \includegraphics[width=0.95\textwidth]{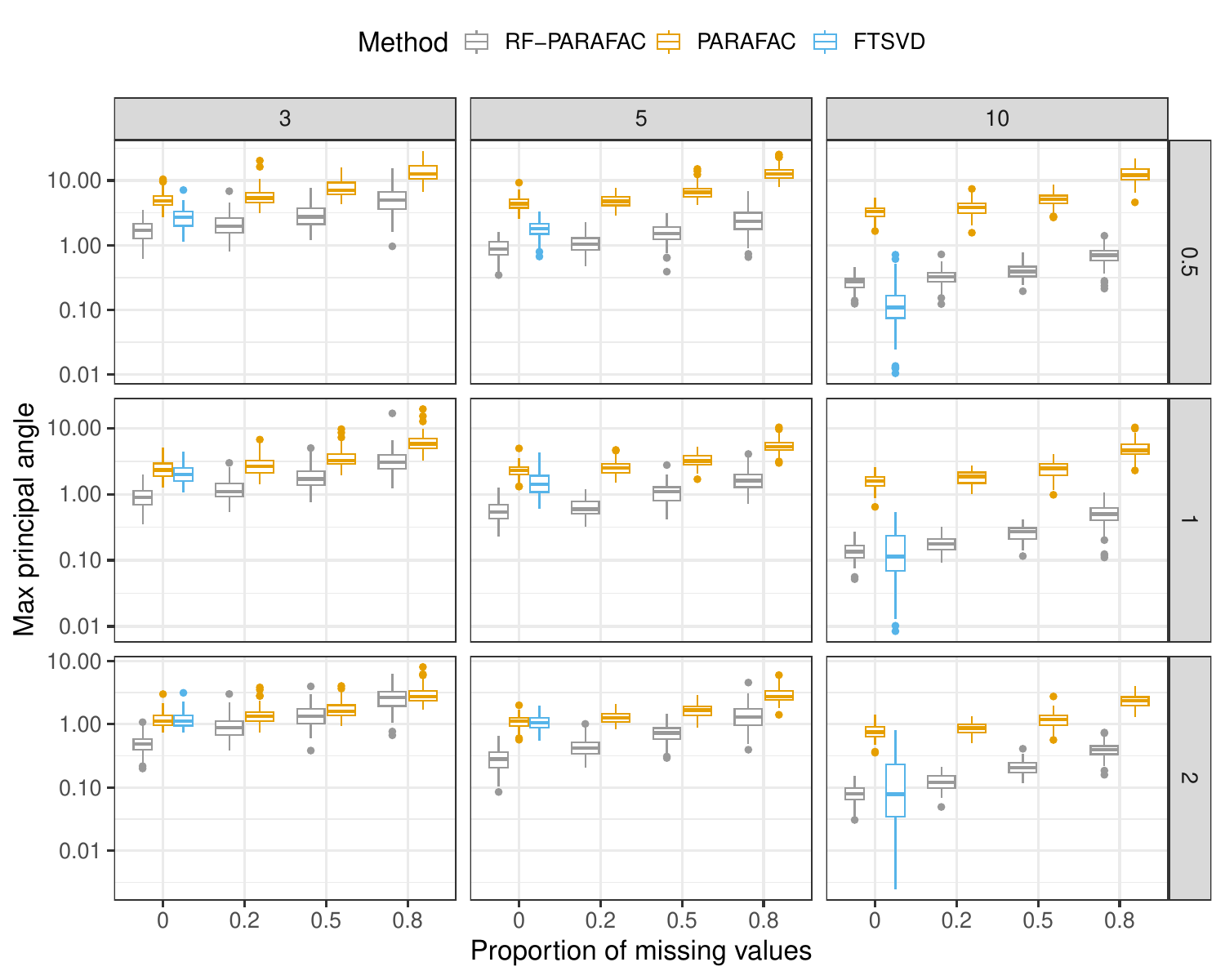}
    \caption{Parameter estimation error (bottom) on order $D=2$ tensors using $p_1 = 10$, with $N=100$ samples. Comparing LF-PARAFAC, PARAFAC (standard), FTSVD, and MFPCA for different values of $R$ (column-wise facets), different proportions of missing values (x-axis), and different signal-to-noise ratios (SNR) (row-wise facets). Plot obtained from $100$ simulation runs.}
    \label{fig:sim}
\end{figure}

As observed on Figure \ref{fig:sim}, we note that our approach outperforms all methods for SNR $\in \{ 0.5, 1 \}$ for the reconstruction task,. For SNR$=2$, the native CPD seems to be the best approach. However, as seen in the parameter estimation task, smoothing approaches are more accurate for retrieving functions $\phi_{r}$ in all cases. Thanks to the smoothing procedure, our approach is particularly efficient when the SNR is low, and the sparsity is high. The FTSVD performs well in retrieving the true feature functions but seems to be slightly worse than the CPD at reconstructing the original tensor. The unfavorable setting may explain this observation, as we do not assume any additional regularity conditions on the space on which lie functions $\phi_{r}$. Nevertheless, the good performance of our approach comes with an overall increased computational cost, as other approaches ran much faster most of the time. 

In the higher-order setting (see Figure \ref{fig:sim2}), we observe overall the same results. One important difference is that all tensor methods outperform the MFPCA in most cases, as the vectorization leads to the tensor structure's loss and is inefficient.

\section{Discussion}
\label{section:discussion}

We introduced a new approach for decomposing efficiently a functional tensor with an underlying random structure. Thanks to the (cross-)covariance formulation of the solving procedure, the method can be applied to a broad range of longitudinal settings,  notably to sparse and irregular sampling schemes. Simulation results show that our approach retrieves feature functions closer to the true smooth functions, and the approximation made from those functions is indeed better.

The developments introduced in this paper pave the way for a broad range of extensions, the more prominent of which is tensor regression. A similar methodology is currently being investigated to perform linear functional tensor regression. Another interesting extension of our work would be to consider that not one but multiple modes have a smooth structure. This situation is often encountered in spatial statistics. As evoked earlier, various tensor decomposition frameworks have been introduced in such settings. However, those methods were developed from a signal-processing setting rather than a statistical setting. Therefore, no probabilistic modeling is made in the decomposition, and the random structure is not considered.

Additionally, our work emphasizes the need for multivariate high-dimensional fast covariance estimation methods. Indeed, one major drawback of our approach is its computational complexity. Currently, the method requires estimating the (cross-)covariance surface between each pair of functional entries of the tensor, a number that grows exponentially with the order of the tensor. We tried to bypass this problem by assuming a separable covariance structure on the covariance function $\boldsymbol{\Sigma}$ since, under simple assumptions, assuming such a structure can significantly reduce the number of computations. Although this approach seems to give good results when the covariance is separable (notably, in simulation settings), it seems not to perform well as soon as the assumption is not verified. 

\section*{Acknowledgements}


Data collection and sharing for the Alzheimer's Disease Neuroimaging Initiative (ADNI) is funded by the National Institute on Aging (National Institutes of Health Grant U19 AG024904). The grantee organization is the Northern California Institute for Research and Education. In the past, ADNI has also received funding from the National Institute of Biomedical Imaging and Bioengineering, the Canadian Institutes of Health Research, and private sector contributions through the Foundation for the National Institutes of Health (FNIH) including generous contributions from the following: AbbVie, Alzheimer’s Association; Alzheimer’s Drug Discovery Foundation; Araclon Biotech; BioClinica, Inc.; Biogen; Bristol-Myers Squibb Company; CereSpir, Inc.; Cogstate; Eisai Inc.; Elan Pharmaceuticals, Inc.; Eli Lilly and Company; EuroImmun; F. Hoffmann-La Roche Ltd and its affiliated company Genentech, Inc.; Fujirebio; GE Healthcare; IXICO Ltd.; Janssen Alzheimer Immunotherapy Research \& Development, LLC.; Johnson \& Johnson Pharmaceutical Research \& Development LLC.; Lumosity; Lundbeck; Merck \& Co., Inc.; Meso Scale Diagnostics, LLC.; NeuroRx Research; Neurotrack Technologies; Novartis Pharmaceuticals Corporation; Pfizer Inc.; Piramal Imaging; Servier; Takeda Pharmaceutical Company; and Transition Therapeutics.

\clearpage

\appendix

\section{Proofs}
\label{sec:annex:proofs}

\begin{proof}[Proposition \ref{proposition:psi}]
    To obtain $\Psi^*$, we look for solutions of
    \begin{equation*}
        \min_{\bold{U}}\| \mathcal{X} - \llbracket \bold{U}; \boldsymbol{\Phi}; \bold{A}_1; \dots; \bold{A}_D \rrbracket\|^2
    \end{equation*}
    for fixed $\boldsymbol{\Phi}, \bold{A}_1, \dots, \bold{A}_D$ and a given $\mathcal{X}$. Developing the criterion, we get
    \begin{align*}
        \| \mathcal{X} - \llbracket \bold{U}; \boldsymbol{\Phi}; \bold{A}_1; \dots; \bold{A}_D \rrbracket\|^2 &= \int \|   \bold{x}(t) - (\bold{A}_D \odot \dots \odot \bold{A}_1 \odot \boldsymbol{\Phi}(t)) \bold{U}^\top \|^2 \\
        &= \int \|   \bold{x}(t) - (\bold{A}_{(D)} \odot \boldsymbol{\Phi}(t)) \bold{U}^\top \|^2
    \end{align*}
    Deriving this expression with respect to U and setting the derivative to $0$ leads to
    \begin{equation*}
        - \int \bold{x}(t)^\top (\bold{A}_{(D)} \odot \boldsymbol{\Phi}(t)) + \bold{U}^* \left ( \int   (\bold{A}_{(D)} \odot \boldsymbol{\Phi}(t))^\top (\bold{A}_{(D)} \odot \boldsymbol{\Phi}(t)) \right ) = 0
    \end{equation*}
    Which, under mild assumptions leads to
    \begin{equation*}
        \bold{U}^* = \left ( \int \bold{x}(t)^T (\bold{A}_{(D)} \odot \boldsymbol{\Phi}(t)) \right ) \left ( \int (\bold{A}_{(D)} \odot \boldsymbol{\Phi}(t))^T (\bold{A}_{(D)} \odot \boldsymbol{\Phi}(t)) \right )^{-1}
    \end{equation*}
\end{proof}


\begin{proof}[Corollary \ref{corollary:cov}]
    From (\ref{eq:solz}) we have
    \begin{align*}
        \mathbb{E}[{\bold{U}^*}^T \bold{U}^*] &= \mathbb{E}[\left ( \int \bold{x}(s)^T \bold{K}(s) \right )^\top \left ( \int \bold{x}(t)^T \bold{K}(t) \right )] \\
        &=   \int \int \bold{K}(s)^\top \mathbb{E}[\bold{x}(s) \bold{x}(t)^T] \bold{K}(t) \\
        &=   \int \int \bold{K}(s)^\top \boldsymbol{\Sigma}_{[f]}(s, t) \bold{K}(t)
    \end{align*}
\end{proof}

\begin{proof}[Lemma \ref{lemma:1}]
    We have
    \begin{align*}
        \mathbb{E} [\| \mathcal{X} - \llbracket \bold{U}; \boldsymbol{\Phi}; \bold{A}_1; \dots; \bold{A}_D \rrbracket\|^2] &= \int \mathbb{E} [\|   \bold{x}(t) - (\bold{A}_D \odot \dots \odot \bold{A}_1 \odot \bold{U}^*) \boldsymbol{\Phi}(t)^\top \|^2] \\
        &= \int \mathbb{E} [ \| \bold{x}(t) - (\bold{A}_{(D)} \odot \bold{U}^*) \boldsymbol{\Phi}(t)^\top \|^2 ] \\
        &= \int \mathbb{E} [\|   \bold{x}(t) \|^2]  - 2 \int \mathbb{E} [\bold{x}(t)^\top (\bold{A}_{(D)} \odot \bold{U}^*) \boldsymbol{\Phi}(t)^\top] \\ & + \int \mathbb{E} [ \|  (\bold{A}_{(D)} \odot \bold{U}^*) \boldsymbol{\Phi}(t)^\top \|^2] \\
    \end{align*}
    In this expression, the first term does not depend on the parameters, and can thus be removed for optimization. Next, using the property $(\bold{A} \otimes \bold{B}) (\bold{C} \odot \bold{D}) = (\bold{A}\bold{C}) \odot (\bold{B}\bold{D})$, we can write
    \begin{align*}
        \int \mathbb{E} [\bold{x}(t)^\top (\bold{A}_{(D)} \odot \bold{U}^*) \boldsymbol{\Phi}(t)^\top] 
        &= \int \mathbb{E} [\bold{x}(t)^\top (\bold{A}_{(D)} \odot \int \bold{x}(t)^\top \bold{K}(s)) \boldsymbol{\Phi}(t)^\top] \\
        &= \int \int \mathbb{E} [\bold{x}(t)^\top (\bold{A}_{(D)} \odot \bold{x}(t)^\top \bold{K}(s)) \boldsymbol{\Phi}(t)^\top] \\
        &= \int \int \mathbb{E} [\bold{x}(t)^\top (\bold{I}_{p_{(D)}} \otimes \bold{x}(t)^\top)] (\bold{A}_{(D)} \odot \bold{K}(s)) \boldsymbol{\Phi}(t)^\top \\
        &= \int \int \boldsymbol{\Sigma}_{[1]}(s, t) (\bold{A}_{(D)} \odot \bold{K}(s)) \boldsymbol{\Phi}(t)^\top \\
    \end{align*}
    and, 
    \begin{align*}
        \int \mathbb{E} [ \|  (\bold{A}_{(D)} \odot \bold{U}^*) \boldsymbol{\Phi}(t)^\top \|^2] &= \int \mathbb{E} [\boldsymbol{\Phi}(t) (\bold{A}_{(D)} \odot \bold{U}^*)^\top (\bold{A}_{(D)} \odot \bold{U}^*) \boldsymbol{\Phi}(t)^\top] \\
        &= \int \mathbb{E} [\boldsymbol{\Phi}(t) (\bold{A}_{(D)}^\top\bold{A}_{(D)} \ast {\bold{U}^*}^\top {\bold{U}^*})\boldsymbol{\Phi}(t)^\top] \\
        &= \int \boldsymbol{\Phi}(t) (\bold{A}_{(D)}^\top\bold{A}_{(D)} \ast \mathbb{E} [{\bold{U}^*}^\top {\bold{U}^*}])\boldsymbol{\Phi}(t)^\top \\
        &=\int \boldsymbol{\Phi}(t) (\bold{A}_{(D)}^\top\bold{A}_{(D)} \ast \boldsymbol{\Lambda}^*)\boldsymbol{\Phi}(t)^\top \\
    \end{align*}
    which leads to 
    \begin{equation}
    \begin{split}
    \label{eq:C}
        C(\boldsymbol{\Lambda}^*, \boldsymbol{\Phi}, \bold{A}_1, \dots, \bold{A}_D) = \int \boldsymbol{\Phi}(t) (\bold{A}_{(D)}^\top \bold{A}_{(D)} \ast \boldsymbol{\Lambda}^*) \boldsymbol{\Phi}(t)^\top \\ -2 \int \int \boldsymbol{\Sigma}_{[f]}(s, t) (\bold{A}_{(D)} \odot \bold{K}(s)) \boldsymbol{\Phi}(t)^\top 
    \end{split}
    \end{equation}
    Alternatively, for any $d\in [D]$, we can rewrite (\ref{eq:C}) as
    \begin{equation*}
        \begin{split}
            C(\boldsymbol{\Lambda}^*, \boldsymbol{\Phi}, \bold{A}_1, \dots, \bold{A}_D) = \int \bold{A}_{d} (\bold{A}_{(-d), f}(t)^\top \bold{A}_{(-d), f}(t) \ast \boldsymbol{\Lambda}^*) \bold{A}_{d}^\top \\ -2 \int \int \boldsymbol{\Sigma}_{[d]}(s, t) (\bold{A}_{(-d), f}(t) \odot \bold{K}(s)) \bold{A}_{d}^\top 
    \end{split}
    \end{equation*}
\end{proof}

In the following we denote $C_f: \boldsymbol{\Phi} \mapsto C(\boldsymbol{\Lambda}^*, \boldsymbol{\Phi}, \bold{A}_1, \dots, \bold{A}_D)$ for any fixed $\boldsymbol{\Lambda}^*, \bold{A}_1, \dots, \bold{A}_D$, and similarly $C_d: \bold{A}_d \mapsto C(\boldsymbol{\Lambda}^*, \boldsymbol{\Phi}, \bold{A}_1, \dots, \bold{A}_D)$ for fixed $\boldsymbol{\Lambda}^*,  \boldsymbol{\Phi}, \bold{A}_1, \dots, \bold{A}_{d-1}, \bold{A}_{d+1}, \dots, \bold{A}_{D}$.

\begin{lemma} 
\label{lemma:fgateaudiff}
    The operator $C_f$ is Gâteau-differentiable and at any $\boldsymbol{\Phi} \in \mathcal{H}^{1\times R}$ and $s \in \mathcal{I}$, the derivative is given by
    \begin{equation*}
        C_f'(\boldsymbol{\Phi})(s) 
        = 2 \boldsymbol{\Phi}(s) (\bold{A}_{(D)}^\top \bold{A}_{(D)} \ast \boldsymbol{\Lambda}^*) -2 \int \boldsymbol{\Sigma}_{[f]}(s, t) (\bold{A}_{(D)} \odot \bold{K}(t))
    \end{equation*}
\end{lemma}

\begin{proof}[Lemma \ref{lemma:fgateaudiff}]
Let $\boldsymbol{\Phi}\in \mathcal{H}^{1\times R}$. Let's show that there is a continuous (bounded) linear operator $C_f'(\boldsymbol{\Phi})$ such that  
\begin{equation}
 \forall \bold{H} \in \mathcal{H}^{1\times R}, \,
\langle C_f'(\boldsymbol{\Phi}) , \bold{H} \rangle
=\lim_{\substack{\alpha \to 0 \\ \alpha>0}}\frac{C_f(\boldsymbol{\Phi}+\alpha \bold{H}) - C_f(\boldsymbol{\Phi})}{\alpha}  
\label{eq.def_gateaux}
\end{equation}
We have
\begin{align*}
    C_f(\boldsymbol{\Phi}+\alpha \bold{H}) = \hspace{1mm}&C_f(\boldsymbol{\Phi}) + 2 \alpha \int \boldsymbol{\Phi}(t) (\bold{A}_{(D)}^\top \bold{A}_{(D)} \ast \boldsymbol{\Lambda}^*) \bold{H}(t)^\top  \\\phantom{=}&-2 \alpha \int \int \boldsymbol{\Sigma}_{[f]}(s, t) (\bold{A}_{(D)} \odot \bold{K}(s)) \bold{H}(t)^\top  \\ \phantom{=}&+ \alpha^2 \int \bold{H}(t) (\bold{A}_{(D)}^\top \bold{A}_{(D)} \ast \boldsymbol{\Lambda}^*) \bold{H}(t)^\top
\end{align*}
Therefore 
\begin{align*}
    \lim_{\substack{\alpha \to 0 \\ \alpha>0}}\frac{C_f(\boldsymbol{\Phi}+\alpha \bold{H}) - C_f(\boldsymbol{\Phi})}{\alpha}  = \hspace{1mm}&2\int \boldsymbol{\Phi}(t) (\bold{A}_{(D)}^\top \bold{A}_{(D)} \ast \boldsymbol{\Lambda}^*) \bold{H}(t)^\top \\ \phantom{=}&-2  \int \int \boldsymbol{\Sigma}_{[f]}(s, t) (\bold{A}_{(D)} \odot \bold{K}(s)) \bold{H}(t)^\top
\end{align*}
The term on the right-hand side can be rewritten $\langle C_f'(\boldsymbol{\Phi}) , \bold{H} \rangle$ where $C_f'(\boldsymbol{\Phi})$ is defined for any $s \in \mathcal{I}$ as
\begin{equation*}
    C_f'(\boldsymbol{\Phi})(s) = 2 \boldsymbol{\Phi}(s) (\bold{A}_{(D)}^\top \bold{A}_{(D)} \ast \boldsymbol{\Lambda}^*) -2 \int \boldsymbol{\Sigma}_{[f]}(s, t) (\bold{A}_{(D)} \odot \bold{K}(t))
\end{equation*}
which is bounded under mild conditions.
\end{proof}

\begin{proof}[Proposition \ref{proposition:solutions}]
We have
\begin{itemize}
    \item Using Lemma \ref{lemma:fgateaudiff} we have for any $\boldsymbol{\Phi} \mathcal{H}^{1 \times R}$ and $s \in \mathcal{I}$,
    \begin{align*}
        C_f'(\boldsymbol{\Phi})(s) = 2 \boldsymbol{\Phi}(s) (\bold{A}_{(D)}^\top \bold{A}_{(D)} \ast \boldsymbol{\Lambda}^*) -2 \int \boldsymbol{\Sigma}_{[f]}(s, t) (\bold{A}_{(D)} \odot \bold{K}(t))
    \end{align*}
    Setting the derivative to zero leads to
    \begin{equation*}
        \boldsymbol{\Phi}(s) (\bold{A}_{(D)}^\top \bold{A}_{(D)} \ast \boldsymbol{\Lambda}^*) - \int \boldsymbol{\Sigma}_{[f]}(s, t) (\bold{A}_{(D)} \odot \bold{K}(t)) = 0
    \end{equation*}
    which implies under non singular assumption,
    \begin{equation*}
        \boldsymbol{\Phi}(s) = \int \boldsymbol{\Sigma}_{[f]}(s, t) (\bold{A}_{(D)} \odot \bold{K}(t)) (\bold{A}_{(D)}^\top \bold{A}_{(D)} \ast \boldsymbol{\Lambda}^*)^{-1}
    \end{equation*}
    \item Similarly we have
    \begin{align*}
        C_d'(\bold{A}_d) = 2 \int \bold{A}_{d} (\bold{A}_{(-d), f}(t)^\top \bold{A}_{(-d), f}(t) \ast \boldsymbol{\Lambda}^*) -2 \int \int \boldsymbol{\Sigma}_{[d]}(s, t) (\bold{A}_{(-d), f}(t) \odot \bold{K}(s))
    \end{align*}
    Setting the derivative to zero leads to
    \begin{equation*}
        \bold{A}_d \int \bold{A}_{(-d), f}(t)^\top \bold{A}_{(-d), f}(t) \ast \boldsymbol{\Lambda}^* - \int \int \boldsymbol{\Sigma}_{[d]}(s, t) (\bold{A}_{(-d), f}(t) \odot \bold{K}(s)) = 0
    \end{equation*}
    which implies under non singular assumption,
    \begin{equation*}
        \bold{A}_d = \int \int \boldsymbol{\Sigma}_{[d]}(s, t) (\bold{A}_{(-d), f}(t) \odot \bold{K}(s)) \left (\int \bold{A}_{(-d), f}(t)^\top \bold{A}_{(-d), f}(t) \ast \boldsymbol{\Lambda}^* \right)^{-1}
    \end{equation*}
\end{itemize}
\end{proof}

\begin{lemma}
    \label{lemma:identifiability}
    Consider the following rank-$R$ latent factor CP decomposition model,
    \begin{align}
        \mathcal{X} &= \llbracket \bold{U}; \bold{V}_1; \dots; \bold{V}_D \rrbracket + \mathcal{E} \\
        \bold{U} &\sim \mathcal{N}(0, \boldsymbol{\Sigma})
    \end{align}
    with $\bold{V}_d \in \mathbb{R}^{p_d \times R}$ for any $d =1,\dots, D$, $\boldsymbol{\Sigma} \in \mathbb{R}^{R \times R}$ a positive definite matrix, and elements of $\mathcal{E}$ i.i.d. normally distributed with variance $\sigma^2$. The model is identifiable under the constraints:
    \begin{enumerate}[label=({A\arabic*})]
    \item \label{asp:A2} When deleting any row of $\bold{V}_M = \bold{V}_1 \odot \dots \odot \bold{V}_D$, the resulting matrix always has two distinct submatrices of rank $R$.
    \item \label{asp:A3} The following inequality is verified
    \begin{equation}
        \sum_{d = 1}^D \krank(\bold{V}_d) \geq R + (D - 1)
    \end{equation}
\end{enumerate}
\end{lemma}

\begin{proof}
    The proof is largely based on \cite{Giordani2020} (Section 4.1). Identifiability is defined as the identifiability of the covariance matrix $\boldsymbol{\Sigma}_{\bold{x}}$ of $\bold{x}$. The residual variance $\sigma$ is identifiable since, using the fact that $\rank(\boldsymbol{\Sigma}) = R$, it is defined by the rank $R$ of the decomposition such as:
    \begin{equation}
        \rank(\boldsymbol{\Sigma}_{\bold{x}} - \sigma^2 \bold{I}) = R
    \end{equation}
    Now, since $\boldsymbol{\Sigma}_{\bold{x}} - \sigma^2 \bold{I}= \bold{V}_M \boldsymbol{\Sigma} \bold{V}_M^\top$ assume we have $\bold{V}_1', \dots \bold{V}_D'$ and $\boldsymbol{\Sigma}'$ such that:
    \begin{equation}
        \bold{V}_M \boldsymbol{\Sigma} \bold{V}_M^\top = \bold{V}_M' \boldsymbol{\Sigma}' \bold{V}_M'{}^\top
    \end{equation}
    which is equivalent to:
    \begin{equation}
        \bold{V}_M \boldsymbol{\Sigma}^{1/2} \boldsymbol{\Sigma}^{1/2} \bold{V}_M^\top = \bold{V}_M' \boldsymbol{\Sigma}'{}^{1/2} \boldsymbol{\Sigma}'{}^{1/2} \bold{V}_M'{}^\top
    \end{equation}
    Under condition \ref{asp:A2}, since $\rank(\boldsymbol{\Sigma}) = R$, if any row of $\bold{V}_M \boldsymbol{\Sigma}^{1/2}$ is deleted there remain two distinct submatrices of rank $R$. By considering Theorem 5.1 in \cite{anderson1956statistical}, we known there exists an orthogonal matrix $\bold{M}$ such that $\bold{V}_M \boldsymbol{\Sigma}^{1/2} = \bold{V}_M' \boldsymbol{\Sigma}'{}^{1/2} \bold{M}$. Finally, using condition \ref{asp:A3} and the fact that $\rank(\boldsymbol{\Sigma}) = R$, we have
    \begin{equation}
        \krank(\boldsymbol{\Sigma}^{1/2}) + \sum_{d=1}^D \krank(\bold{V}_d) \geq 2R + (D - 1)
    \end{equation}
    Using the PARAFAC identifiability result provided in \cite{Sidiropoulos2000}, up to permutation and scaling defined by a matrix $\bold{P}$, we can assert that $\bold{V}_d = \bold{V}_M' \bold{P}$ for any $d=1,\dots, D$ and $\boldsymbol{\Sigma}^{1/2} = \bold{M} \boldsymbol{\Sigma}'^{1/2} \bold{P}$, which implies $\boldsymbol{\Sigma} = \bold{P}^\top \boldsymbol{\Sigma}' \bold{P}$. Identifiability is therefore verified.
\end{proof}

\begin{proof}[Proposition \ref{proposition:identifiability}]
    We follow the same steps as in \cite{Han2023}:
    \begin{itemize}
        \item $\mathscr{H}$ is finite dimensional. Denoting $p_{\phi}$ the dimension of $\mathscr{H}$. We can decompose $\phi_{r}$ and $\tilde{\phi}_{r}$ using an orthonormal basis $\{ \psi_k \}_{1 \leq k \leq p_{\phi}}$ of $\mathscr{H}$,
        \begin{equation*}
            \phi_{r} = \sum^{p_{\phi}}_{k=1} c_{kr} \psi_{k}\ \; \; \; \tilde{\phi}_{r} = \sum^{p_{\phi}}_{k=1} \tilde{c}_{kr} \psi_{k}\
        \end{equation*}
        Considering the following latent factor CP decomposition, 
        \begin{equation*}
            \mathcal{B} = \sum^{R}_{r=1} u_{r} (\bold{c}_{r} \circ \bold{a}_{1 r} \circ \dots \circ \bold{a}_{D r}) \; \; \; \tilde{\mathcal{B}} = \sum^{R}_{r=1} \tilde{u_{r}} (\tilde{\bold{c}}_{r} \circ \tilde{\bold{a}}_{1 r} \circ \dots \circ \tilde{\bold{a}}_{D r})
        \end{equation*}
        We can rewrite (\ref{eq:model:lowrankfu})
        \begin{align*}
            \sum_{r=1}^R u_r \phi_{r}(t) (\bold{a}_{1r} \circ \dots \circ \bold{a}_{Dr}) &= \sum_{r=1}^R u_r \sum^{p_{\phi}}_{k=1} c_{kr} \psi_{k}(t) (\bold{a}_{1r} \circ \dots \circ \bold{a}_{Dr})\\ &= \sum^{p_{\phi}}_{k=1} \sum_{r=1}^R u_r  (c_{kr} (\bold{a}_{1r} \circ \dots \circ \bold{a}_{Dr})) \psi_{k}(t) \\ &= \sum^{p_{\phi}}_{k} \mathcal{B}_k \psi_{k}(t)
        \end{align*}
        where $\mathcal{B}_k$ denotes the sub-tensor of $\mathcal{B}$ defined by fixing the first mode (associated with basis coefficients $\bold{C}$) to the $k$th dimension. Similarly, we have
        \begin{equation*}
            \sum_{r=1}^R \tilde{u}_r \tilde{\phi}_{r}(t) (\tilde{\bold{a}}_{1r} \circ \dots \circ \tilde{\bold{a}}_{Dr}) = \sum^{p_{\phi}}_{k=1} \tilde{\mathcal{B}}_k \psi_{k}(t)
        \end{equation*}
        From the equality
        \begin{equation*}
            \sum^{p_{\phi}}_{k=1} \mathcal{B}_k \psi_{k}(t)= \sum^{p_{\phi}}_{k=1} \tilde{\mathcal{B}}_k \psi_{k}(t)
        \end{equation*}
        and the orthonormality of $\{ \psi_k \}_{1 \leq k \leq p_{\phi}}$, we have that $\mathcal{B} = \tilde{\mathcal{B}}$. From the identifiability result given in Lemma \ref{lemma:identifiability}, assuming scaling and order indeterminacies fixed, we have $\boldsymbol{\Theta} = \tilde{\boldsymbol{\Theta}}$ under the assumption that deleting any row of $\bold{C} \odot \bold{A}_1 \odot \dots \odot \bold{A}_D$ results in a matrix having two distinct full-rank submatrices and considering:
        \begin{equation*}
            \krank(\bold{C}) + \sum_{d=1}^D \krank(\bold{A}_d) \geq R + (D - 1)
        \end{equation*}
        \item $\mathscr{H}$ is infinite dimensional. Considering a regular grid of $\mathcal{I}$ with size $K$, we denote $\mathcal{X}^{(K)}$ and $\tilde{\mathcal{X}}{}^{(K)}$, the values of $\mathcal{X}$ and $\tilde{\mathcal{X}}$ respectively at the time points of the grid. Since $\mathcal{X} = \tilde{\mathcal{X}}$, we have $\mathcal{X}^{(K)} = \tilde{\mathcal{X}}{}^{(K)}$. Considering $\boldsymbol{\phi}_r^{(K)} \in \mathbb{R}^K$ the vector of values of $\phi_r$ at the time points of the grid, we have from the previous equality, 
        \begin{equation*}
            \sum_{r=1}^R u_r  (\boldsymbol{\phi}_r^{(K)} \circ \bold{a}_{1r} \circ \dots \circ \bold{a}_{Dr}) = \sum_{r=1}^R \tilde{u}_r (\tilde{\boldsymbol{\phi}}_r{}^{(K)} \circ \tilde{\bold{a}}_{1r} \circ \dots \circ \tilde{\bold{a}}_{Dr})
        \end{equation*}
        which also corresponds to a SupCP model. In this context, for $K$ large enough, we are guaranteed that $\krank(\boldsymbol{\Phi}^{(K)})=\krank(\boldsymbol{\Phi})$. Using identifiability result from Lemma \ref{lemma:identifiability}, we can assert that  $\boldsymbol{\Phi}^{(K)} = \tilde{\boldsymbol{\Phi}}{}^{(K)}$ as long as deleting any row of $(\boldsymbol{\Phi}^{(K)} \odot \bold{A}_1 \odot \dots \odot \bold{A}_D)$ results in a matrix which has two full rank distinct submatrices and
        \begin{equation*}
           \krank(\boldsymbol{\Phi}^{(K)}) + \sum_{d=1}^D \krank(\bold{A}_d) \geq R + (D - 1)
        \end{equation*}
        Finally, we can assert that there exists $K_T$ such that for all $K \geq K_T$ the first condition is necessarily verified and $\krank(\boldsymbol{\Phi}^{(K)}) = R = \krank(\boldsymbol{\Phi})$
    \end{itemize}
\end{proof}

    

\begin{proof}[Proposition \ref{proposition:scores}]  
    Denoting $\boldsymbol{e}_{i}$ the vector of measurement errors $\{ \mathcal{E} _{ik}\}_{1\leq k \leq N_i}$ of sample $i$. To make notations easier to understand we consider that $\bold{U}_i \in \mathbb{R}^R$ (and not $\mathbb{R}^{1\times R}$) We have
    \begin{equation*}
        \bold{y}_{i} = \bold{F}_{i} \bold{U}_i + \boldsymbol{e}_{i}
    \end{equation*}
    From this, we propose to rewrite jointly the observations and the scores, as:
    \begin{equation*}
        \begin{bmatrix}
        \bold{y}_{i}\\
        \bold{U}_{i}\\
        \end{bmatrix} =
        \begin{bmatrix}
        0\\
        0
        \end{bmatrix}
        +
        \begin{bmatrix}
        \bold{F}_{i} & \bold{I}\\
        \bold{I} & 0
        \end{bmatrix}
        \begin{bmatrix}
        \bold{U}_i\\
        \boldsymbol{e}_{i}
        \end{bmatrix}
    \end{equation*}
    Using this expression, and the jointly Gaussian assumptions on $\bold{U}_i$ and $\boldsymbol{e}_i$ we clearly see that $\bold{U}_i$ and $\bold{y}_i$ are also jointly Gaussian with joint law:
    \begin{align*}
       \begin{bmatrix}
        \bold{y}_{i}\\
        \bold{U}_{i}\\
        \end{bmatrix} &\sim
        N \left ( \begin{bmatrix}
        \bold{F}_{i} & \bold{I}\\
        \bold{I} & 0
        \end{bmatrix} \begin{bmatrix}
        \boldsymbol{\mu}\\
        0\\
        \end{bmatrix} ,
        \begin{bmatrix}
        \bold{F}_{i} & \bold{I}\\
        \bold{I} & 0
        \end{bmatrix}
        \begin{bmatrix}
         \boldsymbol{\Lambda} & 0\\
         0 & \sigma^2 \bold{I} 
        \end{bmatrix}
        \begin{bmatrix}
        \bold{F}_{i}^{\top} & \bold{I}\\
        \bold{I} & 0
        \end{bmatrix}
        \right ) \\
        &\sim
        N \left ( \begin{bmatrix}
        \bold{F}_{i} \boldsymbol{\mu}\\
        \boldsymbol{\mu}\\
        \end{bmatrix} ,
        \begin{bmatrix}
        \bold{F}_{i} \boldsymbol{\Lambda} \bold{F}_i^{\top} + \sigma^2 \bold{I} & \bold{F}_i \boldsymbol{\Lambda}\\
        \bold{\Sigma} \bold{F}_i^{\top} & \boldsymbol{\Lambda}
        \end{bmatrix}
        \right )
    \end{align*}
    To obtain the best predictions for the scores, we now consider the conditional distribution of $\bold{U}_i$ and suggest using the conditional expectation as the best predictor. Using the standard formulation of the Gaussian conditional distribution (\cite{Rasmussen2006}) we obtain :
    \begin{align*}
        \mathbb{E}[\bold{U}_i | \boldsymbol{\mathcal{Y}}_i] &= \boldsymbol{\mu} + \boldsymbol{\Lambda} \bold{F}_i^\top (\bold{F}_i \boldsymbol{\Lambda} \bold{F}_i^\top + \sigma^2\bold{I})^{-1} \boldsymbol{\bold{y}}_i \\
        \boldsymbol{\Sigma}_{\bold{U}_i | \boldsymbol{\mathcal{Y}}_i} &= \boldsymbol{\Lambda} - \boldsymbol{\Lambda} \bold{F}_i^\top (\bold{F}_i \boldsymbol{\Lambda} \bold{F}_i^\top + \sigma^2\bold{I})^{-1} \bold{F}_i \boldsymbol{\Lambda} 
    \end{align*}
    
\end{proof}

\clearpage

\section{Additional experiments}
\label{sec:annex:exp}

\subsection{Simulations: higher order}

We propose to compare decomposition methods on order $D=3$ tensors using the simulation setting of Section \ref{section:simulations}. Dimensions of tensors considered are $(30, 5, 5)$. Results are displayed in Figure \ref{fig:sim2}. We still observe a notable advantage of our approach in high and medium noise settings. On another hand MFPCA seems to perform worse than before. We suppose that this is caused by the lack of assumptions on the underlying structure of tensors.

\subsection{Simulations: misspecified model}

We propose here to investigate how decomposition methods perform on order $D=3$ tensors assuming no underlying low-rank structure. In this context, data is generated using a Fourier basis of order $3$ by sampling Fourier coefficients from a multivariate normal distribution for each tensor entry. The dimensions of tensors considered are $(30, 8, 8)$. For tensor decomposition approaches a rank $R=3$ decomposition is considered. Experiments are repeated 10 times. Results are displayed in Figure \ref{fig:sim3}. As expected, we observe that MFPCA performs better than tensor decomposition approaches in dense and low-sparsity settings since the method does not assume any underlying low-rank structure. Interestingly, our method still performs better than MFPCA for medium-sparsity and high-sparsity settings.

\begin{figure}
    \centering
    \includegraphics[scale=0.45]{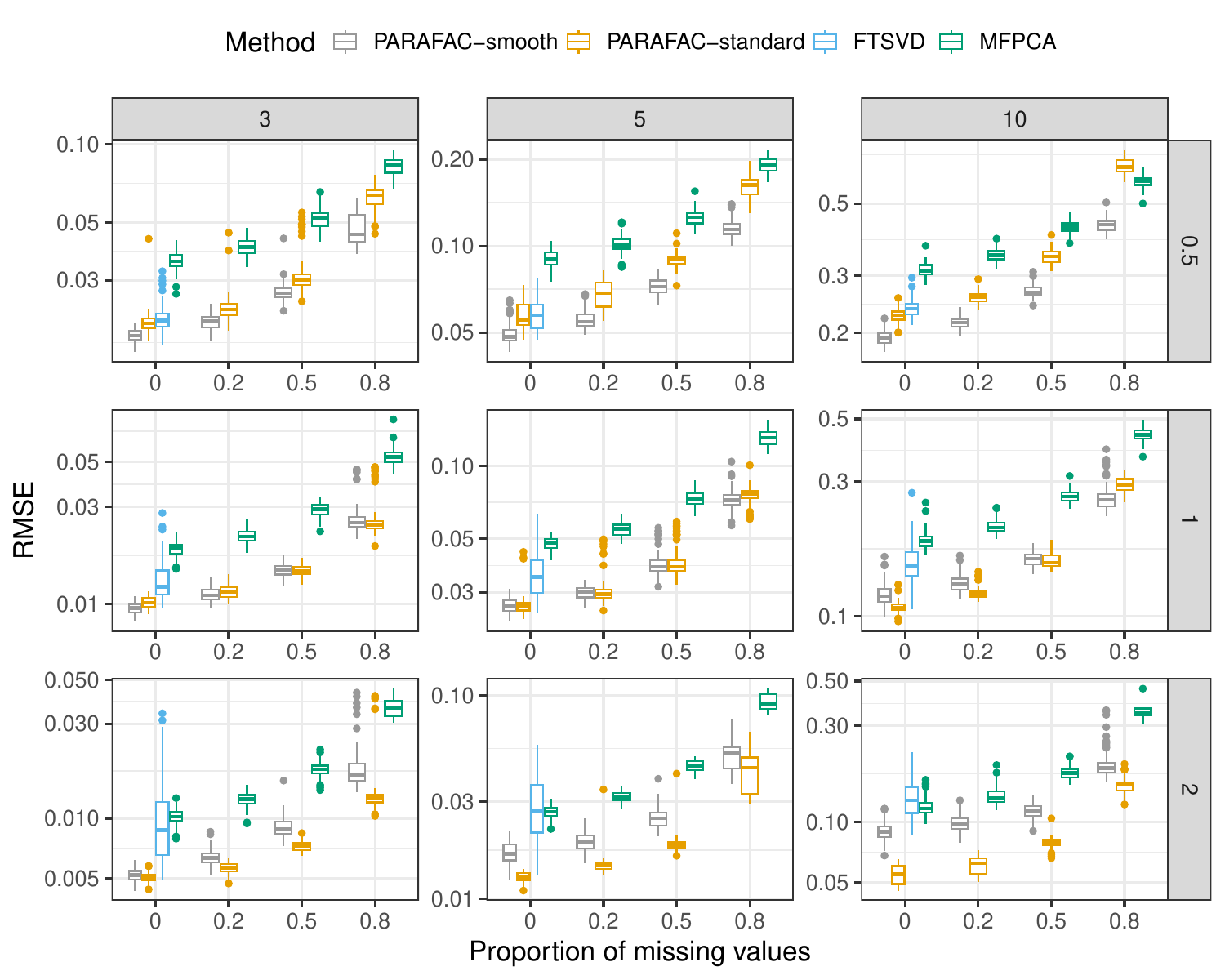}
    \caption{Reconstruction error on order $D=3$ tensors using $p_1 = 10$, with $N=100$ samples. Comparing LF-PARAFAC, PARAFAC (standard), FTSVD, and MFPCA for different values of $R$ (column-wise facets), different proportions of missing values (x-axis), and different signal-to-noise ratios (SNR) (row-wise facets). Plot obtained from $100$ simulation runs.}
    \label{fig:sim2}
\end{figure}

\begin{figure}
    \centering
    \includegraphics[scale=0.45]{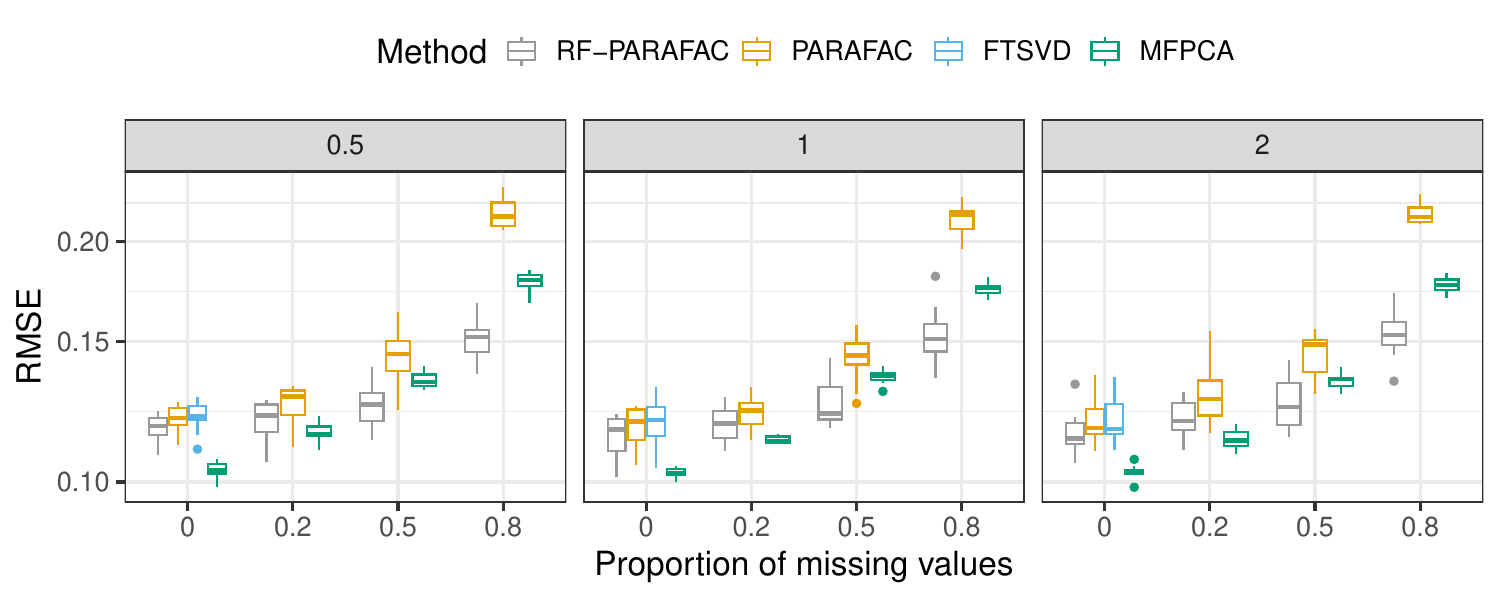}
    \caption{Reconstruction error on order $D=3$ tensors using $p_1=p_2=8$, with $N=100$ samples. Comparing LF-PARAFAC, PARAFAC (standard), FTSVD, and MFPCA for different proportions of missing values (x-axis) and signal-to-noise ratios (SNR) (column-wise facets). Plot obtained from $10$ simulation runs.}
    \label{fig:sim3}
\end{figure}

\begin{figure}
    \centering
    \includegraphics[scale=0.45]{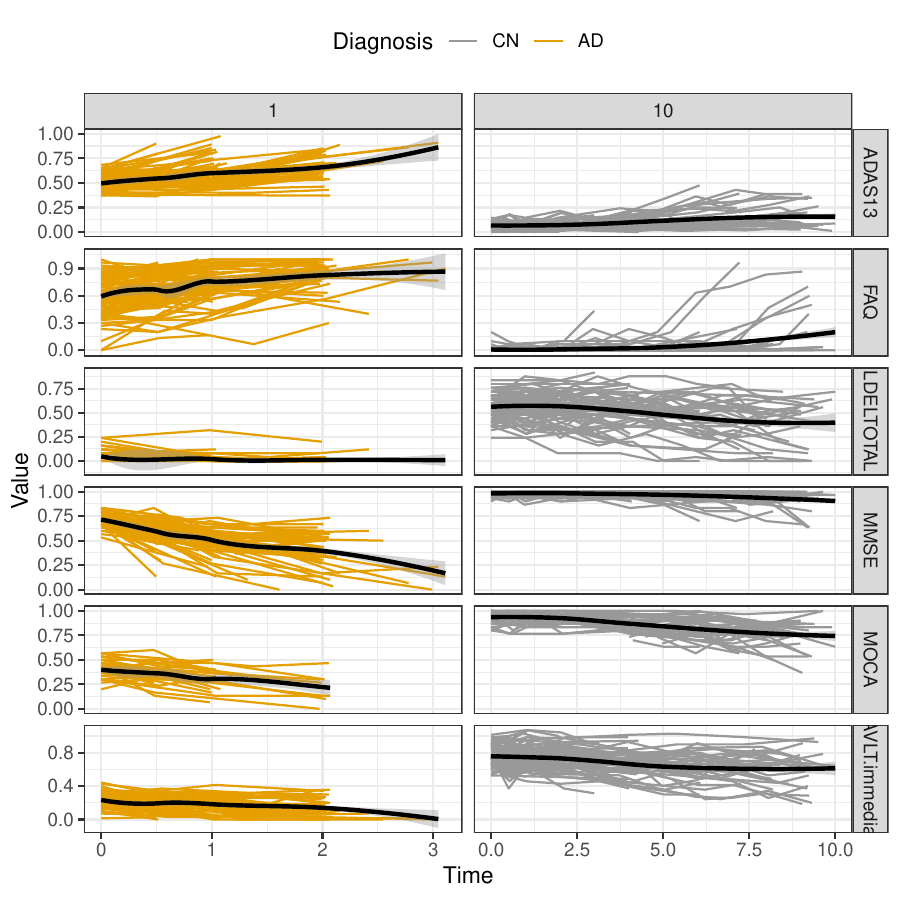}
    \includegraphics[scale=0.45]{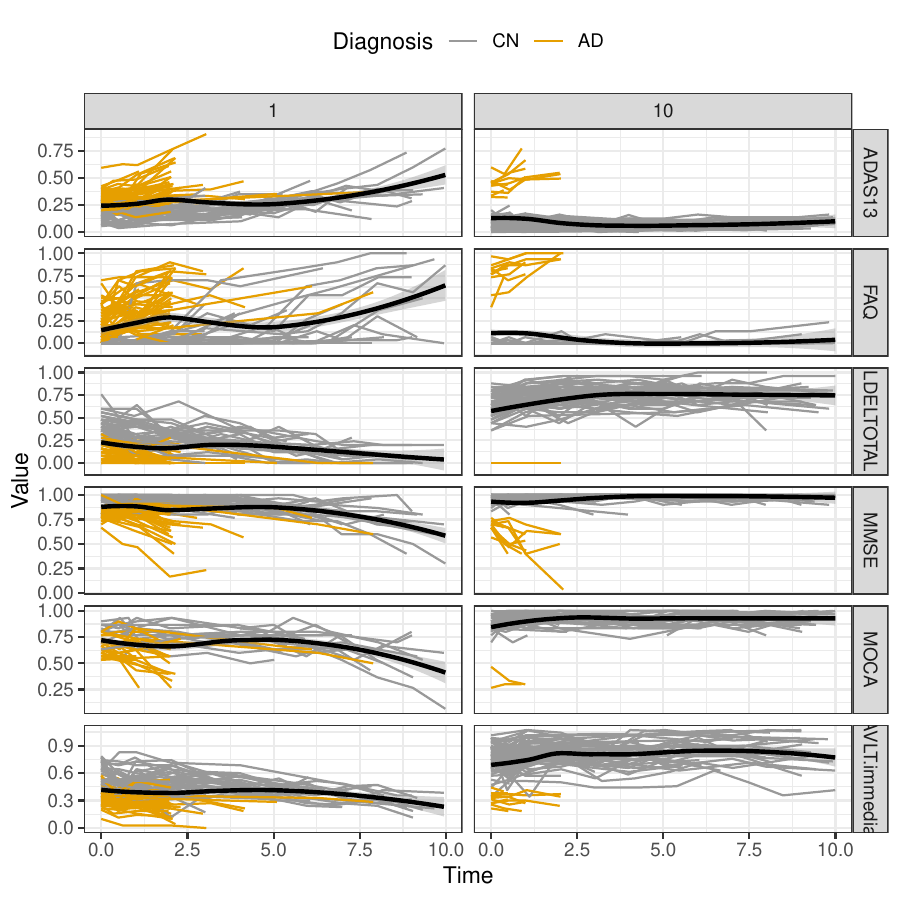}
    \includegraphics[scale=0.45]{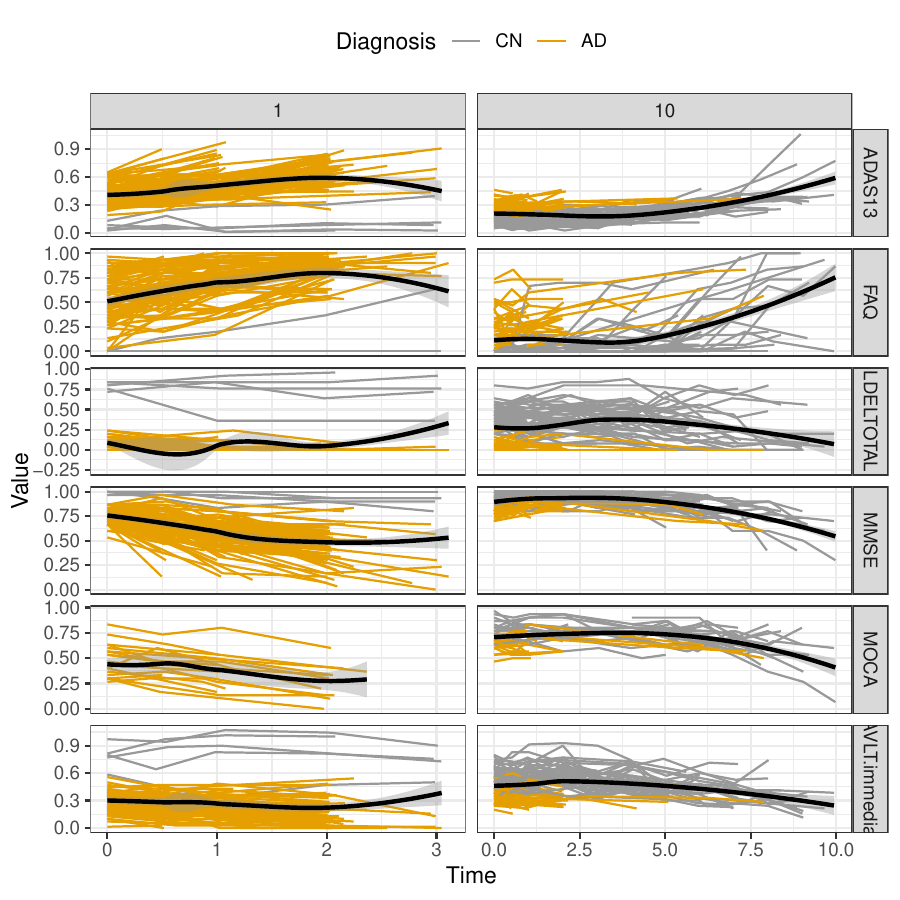}
    \includegraphics[scale=0.45]{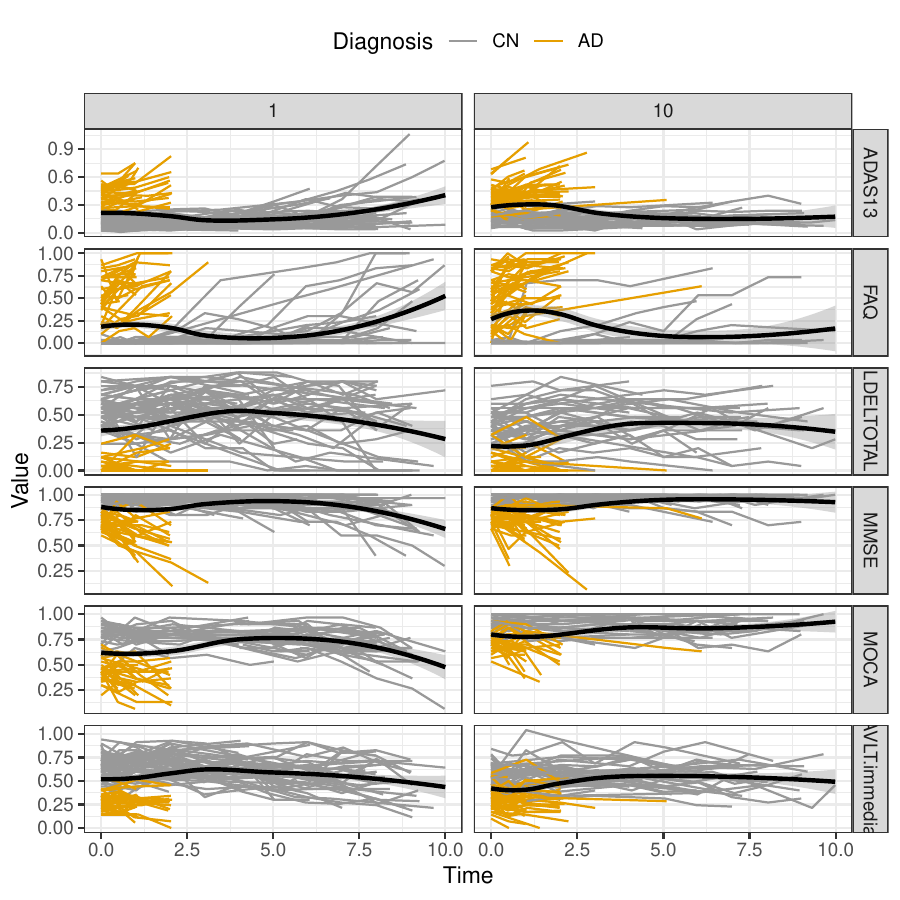}
    \caption{Cognitive score trajectories in first (left column) and last deciles (right column) of rank-1 sample-mode scores (top left), of rank-2 sample-mode scores (top right), of rank-3 sample-mode scores (bottom left), and of rank-4 sample-mode scores (bottom right)}
    \label{fig:sim2}
\end{figure}

\clearpage

\bibliographystyle{abbrvnat}
\bibliography{main}

\end{document}